\documentclass[aps,prl,superscriptaddress,twocolumn,footinbib]{revtex4-1}
\usepackage{amsthm,amsmath,amssymb,amsbsy}

\usepackage{color}
\definecolor{darkred}{rgb}{0.8,0.1,0.1}
\definecolor{lightblue}{rgb}{0.1,0.1,0.8}

\usepackage{graphicx,epic,eepic,epsfig,verbatim,color}
 \usepackage{tikz}
\usepackage[T1]{fontenc}
\usepackage[utf8]{inputenc}
\usepackage[breaklinks]{hyperref}

\newtheorem{theorem}{Theorem}

\newtheorem{corollary}[theorem]{Corollary}

\newtheorem{lemma}[theorem]{Lemma}

\usepackage{cleveref}
\usepackage{graphicx}
\usepackage{xcolor}

\definecolor{darkblue}{RGB}{0,76,156}
\definecolor{darkkblue}{RGB}{0,0,153}
\definecolor{blue2}{RGB}{102,178,255}

\def\endenv{\ifmmode\;\else{\unskip\nobreak\hfil
\penalty50\hskip1em\null\nobreak\hfil\;
\parfillskip=0pt\finalhyphendemerits=0\endgraf}\fi}

\newenvironment{remark}{\noindent \textbf{{Remark~}}}{}
\newenvironment{example}{\noindent \textbf{{Example~}}}{\qed}

\mathchardef\ordinarycolon\mathcode`\:
\mathcode`\:=\string"8000
\def\vcentcolon{\mathrel{\mathop\ordinarycolon}}
\begingroup \catcode`\:=\active
  \lowercase{\endgroup
  \let :\vcentcolon
  }

\makeatletter
\def\resetMathstrut@{%
  \setbox\z@\hbox{%
    \mathchardef\@tempa\mathcode`\[\relax
    \def\@tempb##1"##2##3{\the\textfont"##3\char"}%
    \expandafter\@tempb\meaning\@tempa \relax
  }%
  \ht\Mathstrutbox@\ht\z@ \dp\Mathstrutbox@\dp\z@}
\makeatother
\begingroup
  \catcode`(\active \xdef({\left\string(}
  \catcode`)\active \xdef){\right\string)}
\endgroup
\mathcode`(="8000 \mathcode`)="8000

\newcommand{\nc}{\newcommand}
\nc{\rnc}{\renewcommand}
\nc{\beg}{\begin{equation}}
\nc{\eeq}{{\end{equation}}}
\nc{\beqa}{\begin{eqnarray}}
\nc{\eeqa}{\end{eqnarray}}
\nc{\lbar}[1]{\overline{#1}}
\nc{\ketbra}[2]{|#1\rangle\!\langle#2|}

\nc{\avg}[1]{\langle#1\rangle}
\nc{\Rank}{\operatorname{Rank}}
\nc{\smfrac}[2]{\mbox{$\frac{#1}{#2}$}}
\nc{\tr}{\operatorname{Tr}}
\nc{\ox}{\otimes}
\nc{\dg}{\dagger}
\nc{\dn}{\downarrow}
\nc{\cA}{{\cal A}}
\nc{\cB}{{\cal B}}
\nc{\cC}{{\cal C}}
\nc{\cD}{{\cal D}}
\nc{\cE}{{\cal E}}
\nc{\cF}{{\cal F}}
\nc{\cG}{{\cal G}}
\nc{\cH}{{\cal H}}
\nc{\cI}{{\cal I}}
\nc{\cJ}{{\cal J}}
\nc{\cK}{{\cal K}}
\nc{\cL}{{\cal L}}
\nc{\cM}{{\cal M}}
\nc{\cN}{{\cal N}}
\nc{\cO}{{\cal O}}
\nc{\cP}{{\cal P}}
\nc{\cQ}{{\cal Q}}
\nc{\cR}{{\cal R}}
\nc{\cS}{{\cal S}}
\nc{\cT}{{\cal T}}
\nc{\cV}{{\cal V}}
\nc{\cU}{{\cal U}}
\nc{\cX}{{\cal X}}
\nc{\cY}{{\cal Y}}
\nc{\cZ}{{\cal Z}}
\nc{\cW}{{\cal W}}
\nc{\csupp}{{\operatorname{csupp}}}
\nc{\qsupp}{{\operatorname{qsupp}}}
\nc{\var}{{\operatorname{var}}}
\nc{\rar}{\rightarrow}
\nc{\lrar}{\longrightarrow}
\nc{\polylog}{{\operatorname{polylog}}}
\nc{\1}{{\mathds{1}}}
\nc{\wt}{{\operatorname{wt}}}
\nc{\av}[1]{{\left\langle {#1} \right\rangle}}
\nc{\supp}{{\operatorname{supp}}}

\def\x{\xi}

\nc{\RR}{{{\mathbb R}}}
\nc{\CC}{{{\mathbb C}}}
\nc{\FF}{{{\mathbb F}}}
\nc{\NN}{{{\mathbb N}}}
\nc{\ZZ}{{{\mathbb Z}}}
\nc{\PP}{{{\mathbb P}}}
\nc{\QQ}{{{\mathbb Q}}}
\nc{\UU}{{{\mathbb U}}}
\nc{\EE}{{{\mathbb E}}}
\nc{\CHSH}{{\operatorname{CHSH}}}

\nc{\be}{\begin{equation}}
\nc{\ee}{{\end{equation}}}
\nc{\bea}{\begin{eqnarray}}
\nc{\eea}{\end{eqnarray}}
\nc{\Hom}[2]{\mbox{Hom}(\CC^{#1},\CC^{#2})}
\nc{\rU}{\mbox{U}}

\nc{\ob}[1]{#1}

\nc{\SEP}{{\text{SEP}}}
\nc{\NS}{{\text{NS}}}
\nc{\LOCC}{{\text{LOCC}}}
\nc{\PPT}{{\text{PPT}}}
\nc{\EXT}{{\text{EXT}}}
\nc{\Sym}{{\operatorname{Sym}}}

\nc{\ERLO}{{E_{\text{r,LO}}}}
\nc{\ERLOCC}{{E_{\text{r,LOCC}}}}
\nc{\ERPPT}{{E_{\text{r,PPT}}}}
\nc{\ERLOCCinfty}{{E^{\infty}_{\text{r,LOCC}}}}
\nc{\Aram}{{\operatorname{\sf A}}}

\usepackage{tikz}
\usetikzlibrary{arrows}

\makeatletter
\def\grd@save@target#1{%
  \def\grd@target{#1}}
\def\grd@save@start#1{%
  \def\grd@start{#1}}
\tikzset{
  grid with coordinates/.style={
    to path={%
      \pgfextra{%
        \edef\grd@@target{(\tikztotarget)}%
        \tikz@scan@one@point\grd@save@target\grd@@target\relax
        \edef\grd@@start{(\tikztostart)}%
        \tikz@scan@one@point\grd@save@start\grd@@start\relax
        \draw[minor help lines,magenta] (\tikztostart) grid (\tikztotarget);
        \draw[major help lines] (\tikztostart) grid (\tikztotarget);
        \grd@start
        \pgfmathsetmacro{\grd@xa}{\the\pgf@x/1cm}
        \pgfmathsetmacro{\grd@ya}{\the\pgf@y/1cm}
        \grd@target
        \pgfmathsetmacro{\grd@xb}{\the\pgf@x/1cm}
        \pgfmathsetmacro{\grd@yb}{\the\pgf@y/1cm}
        \pgfmathsetmacro{\grd@xc}{\grd@xa + \pgfkeysvalueof{/tikz/grid with coordinates/major step}}
        \pgfmathsetmacro{\grd@yc}{\grd@ya + \pgfkeysvalueof{/tikz/grid with coordinates/major step}}
        \foreach \x in {\grd@xa,\grd@xc,...,\grd@xb}
        \node[anchor=north] at (\x,\grd@ya) {\pgfmathprintnumber{\x}};
        \foreach \y in {\grd@ya,\grd@yc,...,\grd@yb}
        \node[anchor=east] at (\grd@xa,\y) {\pgfmathprintnumber{\y}};
      }
    }
  },
  minor help lines/.style={
    help lines,
    step=\pgfkeysvalueof{/tikz/grid with coordinates/minor step}
  },
  major help lines/.style={
    help lines,
    line width=\pgfkeysvalueof{/tikz/grid with coordinates/major line width},
    step=\pgfkeysvalueof{/tikz/grid with coordinates/major step}
  },
  grid with coordinates/.cd,
  minor step/.initial=.2,
  major step/.initial=1,
  major line width/.initial=2pt,
}
\makeatother

\tikzset{
  treenode/.style = {align=center, inner sep=0pt, text centered,
    font=\sffamily},
  arn_n/.style = {treenode, circle, white, font=\sffamily\bfseries, draw=black,
    fill=black, text width=1.5em},% arbre rouge noir, noeud noir
  arn_r/.style = {treenode, circle, red, draw=red, 
    text width=1.5em, very thick},% arbre rouge noir, noeud rouge
  arn_x/.style = {treenode, rectangle, draw=black,
    minimum width=0.5em, minimum height=0.5em}% arbre rouge noir, nil
}

\usepackage{pifont}% http://ctan.org/pkg/pifont
\usepackage{graphicx}
\usepackage{bm,bbm}% 
\usepackage{braket}
\usepackage{enumerate}
\usepackage{subcaption,caption}
\usepackage{booktabs,multirow}
\usepackage{mathtools,bbding}
\usepackage[percent]{overpic}
\usepackage{microtype}

\usepackage{newpxtext,newpxmath}

\usepackage[most,breakable]{tcolorbox}

\AtBeginDocument{
\heavyrulewidth=.08em
\lightrulewidth=.05em
\cmidrulewidth=.03em
\belowrulesep=.65ex
\belowbottomsep=0pt
\aboverulesep=.4ex
\abovetopsep=0pt
\cmidrulesep=\doublerulesep
\cmidrulekern=.5em
\defaultaddspace=.5em
}

\usepackage{etoolbox}
\hypersetup{linktoc=all}

\DeclareMathOperator{\Tr}{Tr}

\newcommand{\R}{\mathcal{R}}

\newcommand{\C}{\mathcal{C}}

\nc{\MIO}{{\text{\rm MIO}}}
\nc{\DIO}{{\text{\rm DIO}}}
\nc{\SIO}{{\text{\rm SIO}}}
\nc{\IO}{{\text{\rm IO}}}

\let\oldproofname\proofname
\renewcommand{\proofname}{\rm\bf{\oldproofname}}

\makeatletter
\renewenvironment{proof}[1][\proofname]{%
  \vspace{-\topsep}
  \pushQED{\qed}
  \normalfont
  \topsep6\p@\@plus6\p@\relax
  \trivlist\item[\hskip\labelsep\bfseries#1\@addpunct{.}]\ignorespaces}{\popQED\endtrivlist\@endpefalse}
\makeatother

\begin{document}

\title{Inequivalent Multipartite Coherence Classes and New Coherence Monotones}

%%%%%%%%%%%%%%%%%%%%%%%%
\author{Yu Luo}
\email{luoyu@snnu.edu.cn}
\affiliation{College of Computer Science, Shaanxi Normal University, Xi'an, 710062, China}
\affiliation{Centre for Quantum Software and Information, Faculty of Engineering and Information Technology, University of Technology Sydney, Australia}

\author{Yongming Li}
\email{liyongm@snnu.edu.cn}
\affiliation{College of Computer Science, Shaanxi Normal University, Xi'an, 710062, China}

\author{Min-Hsiu Hsieh}
\email{min-hsiu.hsieh@uts.edu.au}
\affiliation{Centre for Quantum Software and Information, Faculty of Engineering and Information Technology, University of Technology Sydney, Australia}

%%%%%%%%%%%%%%%%%%%%%%%%%

\begin{abstract}
Quantum coherence has received significant attention in recent years, but its study is mostly conducted in single party settings. In this paper, we generalize important results in multipartite entanglement theory to their counterparts in quantum coherence theory. First, we give a necessary and sufficient condition for when two pure multipartite states are equivalent under local quantum incoherent operations and classical communication (LICC), i.e., two states can be deterministically transformed to each other under LICC operations. Next, we investigate and give the conditions in which such a transformation succeeds only stochastically.  Different from entanglement case for two-qubit states, we find that the SLICC equivalence classes are infinite. Moreover, it's possible that there are some classes of states in multipartite entanglement can convert into each other, while, they cannot convert into each other in multipartite coherence. In order to show the difference among SLICC classes, we introduce two coherence monotones: accessible coherence and source coherence, following the logistics given in [\emph{Phys.~Rev.~Lett. 115,~150502 (2015)}]. These coherence monotones have a straightforward operational interpretation, namely, the accessible coherence characterizes the proficiency of a state to generate other states via quantum incoherent operations, while the source coherence characterizes the set of states that can be reached via quantum incoherent operations acting on the given state of interest. 
\end{abstract}

%\date{\today}
\maketitle

\section{Introduction}
Coherence originates from the ``superposition'' of quantum states, and plays a central role in interference phenomena in quantum physics and quantum information science~\cite{Baumgratz14, Shi1705, Streltsov1705,Streltsov1502,Chitambar16,Streltsov_prx,Du15,Shao15,Xi15,Hu17,Winter16,Chitambar16_PIO}. Coherence is an essential ingredient for multipartite entanglement in many-body systems and a necessary phenomenon of analysing physical phenomena in quantum optics\cite{PhysRevLett.94.173602}, solid state physics\cite{Deveaud2009}, and nanoscale thermodynamics\cite{Cwiklinski2014,Lostaglio2015,Karlstrom2011}, even in biological systems\cite{Lloyd2011,Huelga2013,Roden2016}. A mathematical framework of quantum coherence as a physical resource has been proposed recently~\cite{Baumgratz14}. There are two basic elements in coherence theory: (1) free states and (2) free operations. Free states in the coherence theory are those states which are diagonal in a fixed basis $\{\ket{i}\}$, which we call incoherent states. Free operations (incoherent operations) are some specified classes of physically realizable operations that act invariantly on the set of incoherent states, which is not unique due to practical implications.

%Local operations assisted by classical communications (LOCC) is helpful for understanding the structure of entangled states because entanglement cannot be created by LOCC operators. In analogy with entanglement theory, for multipartite coherent states, we focus on the protocol that each party performs local incoherent operations assisted by classical communications (LICC)~\cite{Chitambar16,Streltsov_prx}. In LICC protocol, the local incoherent operator cannot create global coherence so that multipartite coherence remains a resource. LICC has been used in exploring the interplay between coherence and entanglement as resource primitives in quantum information theory~\cite{Chitambar16,Streltsov_prx,Matera2016,Egloff18}. Meanwhile, LICC can be viewed as a natural setting to explore the structure of a multipartite coherent states. 

Local operations assisted by classical communications (LOCC) is helpful for understanding the structure of entangled states because entanglement cannot be created by LOCC operators. In analogy with entanglement theory, we focus on the protocol that each party performs local incoherent operations assisted by classical communications (LICC)~\cite{Chitambar16,Streltsov_prx}. In LICC protocol, the local incoherent operator cannot create global coherence so that multipartite coherence remains a resource. Therefore, LICC can be viewed as a natural setting to explore the structure of a multipartite coherent states. 

The first aim in this paper is to understand the structure of a multipartite coherent states. For quantum information processing, different structures of states often have different capabilities in state transformation. Just like the resource of entanglement, a LOCC protocol for multipartite state leads to natural ways of defining equivalent relations in the set of entangled states, as well as establishing hierarchies between the resulting classes (structures). States in different classes cannot convert into each other. For example, GHZ states and W states belong to two different classes via stochastic local operations and class communications (SLOCC), which reveals the existence of two inequivalent kinds of genuine tripartite entanglement~\cite{Dur00}. To understand the structure of multipartite coherent states, we also consider the classification of pure coherent states of multipartite quantum systems under LICC, even in a stochastic setting. For multipartite coherence, an interesting observation is: if one of the three qubits in W state is lost, the state of the remaining two-qubit system is still coherent, whereas the GHZ state, which is incoherent after the loss of one qubit. This observation leads to following question: is it possible that there are some classes of states in multipartite entanglement can convert into each other, while, they cannot convert into each other in multipartite coherence? 

We address the questions above by focusing on the equivalent class of pure multipartite coherent states. First, a necessary and sufficient condition for pure multipartite state transformations via local incoherent operations and class communications (LICC) is presented. Second, we investigate and give the conditions in which such a transformation succeeds only stochastically, namely, stochastic local incoherent operations and class communications (SLICC). As an application, we investigate the two-qubit SLICC equivalent classes, showing that equivalent classes of bipartite coherence under SLICC are more complex than equivalent classes of bipartite entanglement under SLOCC, even in the two-qubit case. This specific incoherent constraint of locality makes us undertand the structure of multipartite coherence. 

%Because the different structure of multipartite coherent states have different capacities in state transformation task via LICC and many of coherence measures do not have an interpretation in the context of the LICC paradigm~\cite{Streltsov17_rmp}.

%In order to quantify coherence, possible LICC transformations among multipartite coherent states have to be further investigated with the intention to identify new operational coherent measures. 

Because many coherence measures do not have an interpretation in the context of the LICC paradigm~\cite{Streltsov17_rmp}. The another aim in this paper is to identify new operational coherence measures for investigating possible LICC transformations among coherent multipartite states. Inspired by similar concepts previously investigated for entanglement~\cite{Schwaiger15}, we introduce two coherence monotones: accessible coherence and source coherence. Accessible coherence refers to the proficiency of a state to generate other states via free operations, and the source coherence denotes that the set of a state of interest can be obtained via free operations. Both coherence monotones can also be used for single-qubit and multipartite state cases and can be applied for other free operations, such as incoherent operations (IC), LICC. For single-qubit states achieved via physically incoherent operations (PIO)~\cite{Chitambar16_PIO}, strictly incoherent operations (SIO)~\cite{Winter16} and IC, we obtain explicit formulas of accessible coherence and source coherence. We analyze pure (or mixed) states via IC and derive explicit formulas for the source coherence. These two new coherence monotones also have a geometric interpretation. We additionally show how accessible coherence can be computed numerically and provide examples.

%%%%%%%%%%%%%%%%%%%%%%%%%%%%%%%%%%%%%%%%%%%%
%%%%%%%%%%%%%%%%%%%%%%%%%%%%%%%%%%%%%%%%%%%%
%%%%%%%%%%%%%%%%%%%%%%%%%%%%%%%%%%%%%%%%%%%%

\section{Inequivalent Classes of Multipartite Coherence States}\label{sec:multiclass}
In this section, we will explore inequivalent classes of multipartite coherence states. First, we will give a necessary and sufficient condition for when two multipartite coherence states can be interconverted with certainty under local incoherent operations and classical communication (LICC). We then study the interconversion of multipartite coherence states which only succeeds with a strictly positive probability. This allows us to define inequivalent classes of multipartite coherence states. The discussion in this section closely follows  the inequivalent classes of multipartite entanglement states in \cite{Bennett00, Dur00}.

We start with the following lemma showing that the number of product terms of a multipartite coherent pure state is an SLICC monotone. Note that the number of product terms has been shown by other authors to be IC monotone ~\cite{Winter16,Killoran16}. We extend this idea and show that the number of product terms is also an SLICC monotone for multipartite coherence. We say that two pure states $\ket{\psi}$ and $\ket{\phi}$ are $\mathcal{O}$ equivalent if they can be transformed into each other by means of operations in the set $\mathcal{O}$.  
\begin{lemma}[LICC and SLICC monotone]\label{lemma_corank}
\label{lem:product}
The number of non-zero product terms in the fixed basis does not increase under LICC (resp.~SLICC). 
\end{lemma}

The proof can be found in Supplemental Material \ref{sec:proof_corank}.

\begin{remark}
The number of product terms in the fixed basis is similar to Schmidt rank in the resource theory of entanglement, but they do not play the same roles. Schmidt rank, being an entanglement monotone under SLOCC, can be used to classify the structures of entanglement, and this classification is complete for bipartite settings. In other words, two pure entangled states are SLOCC equivalent iff they have the same Schmidt rank.  However, the number of product terms alone, despite being an SLICC monotone, is not sufficient to classify all the structures of multipartite pure coherent states. Example~\ref{ex:SLICC2} below shows that two pure states can be SLICC inequivalent even though they have the same number of product terms.
\end{remark}

Our first main result is the following, which originates from its entanglement counterpart in \cite[Corollary 1]{Bennett00}.

\begin{theorem}\label{The_LICCeq}
Two multipartite pure states $\ket{\psi}$ and $\ket{\phi}$ are LICC equivalent iff they are local IU (LIU) equivalent.
\end{theorem}

The proof can be found in Supplemental Material \ref{sec:proof_LICCeq}.

It is also possible that two multipartite pure states $\ket{\psi}$ and $\ket{\phi}$ cannot always succeed with certainty in interconverting through operations in the class $\mathcal{O}$, i.e., such a transformation may only succeed stochastically. This allows the structure of multipartite coherence states to be understood operationally (similar to that of multipartite entangled states): if two multipartite pure states $\ket{\psi}$ and $\ket{\phi}$ cannot be transformed to each other with non-zero probability, they must each belong to different types of multipartite coherence structures.

In the remainder of this section, we will focus on $SLICC$. Next we give our main result:

\begin{theorem}\label{The_LSICC}
Two multipartite pure states $\ket{\psi}$ and $\ket{\phi}$ are equivalent under SLICC if and only if they are related by local (invertible) SIO operators.
\end{theorem}

The proof can be found in Supplemental Material \ref{sec:proof_SLICC}.

To conclude this section, we use an example to demonstrate that the $N$-partite coherence states are already more versatile  than entanglement when $N=2$.

\begin{example}\label{ex:SLICC2} [Characterization of two-qubit coherence states]
Consider a two-qubit system with the fixed basis $\{ \ket{00},\ket{01},\ket{10},\ket{11}\}$. Our classification is based on the number of product terms $R$ in a two-qubit pure state since this number will not be altered by invertible SIO operators. The following table lists all inequivalent classes of two-qubit states. 

\begin{center}
\begin{tabular} { | c | c | c | p{2.5cm} |} 
\hline  
R & Classification \\ \hline  
1 & $\ket{00}$ \\ \hline 
2 & $a\ket{00}+b\ket{01}$, $a\ket{00}+c\ket{10}$, $a\ket{00}+d\ket{11}$\\ \hline 
3 & $a\ket{00}+b\ket{01}+c\ket{10}$, $a\ket{00}+b\ket{01}+d\ket{11}$\\ \hline  
4 & infinitely many (based on different $\mathcal{r}$)\\ \hline  
\end{tabular}
\end{center}
\end{example}

The detailed analysis can be found in Supplemental Material \ref{sec:proof_ex_SLICC2}. We can show that it's possible that there are some classes of states in multipartite entanglement can convert into each other, while, they cannot convert into each other in multipartite coherence even in two-qubit case.

From the example above, we can see SLICC equivalent class of $\ket{\psi}=a\ket{00}+b\ket{01}+c\ket{10}+d\ket{11}$ can form an one-parameter family of states. 
\begin{corollary} \label{Cor_SLICC}
Any state $\ket{\psi}=a\ket{00}+b\ket{01}+c\ket{10}+d\ket{11}$ which the number of product terms equals 4 is SLICC equivalent to a state with form $\ket{\psi'}=\alpha(\ket{00}+\ket{01}+\ket{10})+\beta\ket{11}$, where $\beta$ is any complex number with $0<|\beta|<1$, and $\alpha$ is the real number determined by normalization. That is, SLICC equivalent class of $\ket{\psi}=a\ket{00}+b\ket{01}+c\ket{10}+d\ket{11}$ can form an one-parameter family of states. 
\end{corollary}

The proof can be found in Supplemental Material \ref{sec:proof_Cor_SLICC}.

%\begin{proof}
%As discussed above, we find the following 4 operators
%\begin{eqnarray}
%\begin{pmatrix} \frac{\alpha}{b\beta} & 0   \\ 0 & \frac{1}{d} \end{pmatrix}\otimes\begin{pmatrix} \frac{d\alpha}{c} & 0   \\ 0 & \beta \end{pmatrix},
%\end{eqnarray}
%\begin{eqnarray}
%\begin{pmatrix} \frac{\alpha}{b\beta} & 0   \\ 0 & \frac{1}{c} \end{pmatrix}\otimes\begin{pmatrix} 0 & \frac{c\alpha}{d}   \\ \beta & 0 \end{pmatrix},
%\end{eqnarray}
%\begin{eqnarray}
%\begin{pmatrix} 0 & \frac{\alpha}{b\beta}   \\ \frac{1}{b} & 0 \end{pmatrix}\otimes\begin{pmatrix} \frac{b\alpha}{a} & 0   \\ 0 & \beta \end{pmatrix},
%\end{eqnarray}
%\begin{eqnarray}
%\begin{pmatrix} 0 & \frac{\alpha}{c\beta}   \\ \frac{1}{a} & 0 \end{pmatrix}\otimes\begin{pmatrix} 0 & \frac{a\alpha}{b}   \\ \beta & 0 \end{pmatrix},
%\end{eqnarray}
%can transform $\ket{\psi}=a\ket{00}+b\ket{01}+c\ket{10}+d\ket{11}$ to $\ket{\psi'}=\alpha(\ket{00}+\ket{01}+\ket{10})+\beta\ket{11}$, where $0<|\beta|<1$ and $\alpha=\frac{\sqrt{1-|\beta|^2}}{3}>0$. Thus, the state $\ket{\psi}=a\ket{00}+b\ket{01}+c\ket{10}+d\ket{11}$ is SLICC equivalent to the state $\ket{\psi'}=\alpha(\ket{00}+\ket{01}+\ket{10})+\beta\ket{11}$ via those operators above. 
%\end{proof}

\begin{remark}\label{re:two-qubits}
In entanglement theory, the degree of entanglement can be measured by concurrence $C_E$~\cite{Wootters98}. If we write $\ket{\psi}=a\ket{00}+b\ket{01}+c\ket{10}+d\ket{11}$, then $C_{E}(\ket{\psi})=2|ad-bc|$~\cite{Wootters98}. Suppose that $\ket{\psi}$ and $\ket{\phi}$ are SLOCC equivalent, then there exist local invertible operators $A$ and $B$, such that $\ket{\phi}=A\otimes B\ket{\psi}$~\cite{Dur00}. Consequently, $C_{E}(\ket{\phi})=det(A)det(B)C_{E}(\ket{\psi})$. There are only two SLOCC equivalent classes: $C_{E}(\ket{\psi})=0$ and $C_{E}(\ket{\psi})\neq0$. Equivalently, the classification depends on either  $ad=bc$ or $ad\neq bc$. The classification of SLICC is different from the classification of SLOCC in entanglement theory. As shown in the two-qubit case, SLICC classification depends both on the number of product terms and the number $\mathcal{r}=\frac{ad}{bc}$. Because SLICC equivalent class of $\ket{\psi}=a\ket{00}+b\ket{01}+c\ket{10}+d\ket{11}$ can form an one-parameter family of states, we can further simplify $\mathcal{r}$ as $\mathcal{r}=\frac{\beta}{\alpha}=\frac{3\beta}{\sqrt{1-|\beta|^2}}$.
\end{remark}

%\begin{figure}
%\centering
%\includegraphics[width=0.48\textwidth]{slicc2qubit.eps}
%\caption{(color online). As shown in this figure, we have plotted two SLICC inequivalent classes for $\mathcal{r}=9$ (red) and $\mathcal{r}=1$ (green). We have plotted the edges of the class of $\mathcal{r}=9$ (yellow) and shown that no intersection of these two classes exists.  }
%\label{SLICC2qubit}
%\end{figure}

%%%%%%%%%%%%%%%%%%%%%%%%%%%%%%%%%%%%%%%%%%%%%%%%%%%%
%%%%%%%%%%%%%%%%%%%%%%%%%%%%%%%%%%%%%%%%%%%%%%%%%%%%
%%%%%%%%%%%%%%%%%%%%%%%%%%%%%%%%%%%%%%%%%%%%%%%%%%%%

\section{New operational coherence monotones: Accessible coherence and Source coherence}\label{sec:ASC}

In this section, we will recall the framework for quantifying the resource of coherence theory and then introduce two new operational coherence monotones: accessible coherence and source coherence. Our idea comes from Schwaiger \emph{et al.} \cite{Schwaiger15} and Sauerwein \emph{et al.} \cite{Sauerwein15}, in which the authors studied similar entanglement measures: accessible entanglement and source entanglement.  

Baumgratz $et$ $al.$~\cite{Baumgratz14} proposed a seminal framework for quantifying coherence as a resource. For a fixed basis $\{\ket{i}\}$, a functional $C$ can be taken as a coherence measure if it satisfies the following four conditions:

$(B1)$ $C(\rho)\geq0$ for all quantum states, and $C(\rho)=0$ if $\rho\in\mathcal{I}$, where $\mathcal{I}$ is the set of incoherence states which are diagonal in basis $\{\ket{i}\}$;

$(B2)$ $C(\rho)\geq C(\Phi(\rho))$ for all free operations $\Phi$;

$(B3)$ $C(\rho)\geq \sum_np_nC(\rho_n)$, where $p_n=\Tr(K_n\rho K_n^\dag)$, $\rho_n=\frac{1}{p_n}K_n\rho K_n^\dag$ and $K_n$ are the Kraus operators of an incoherent $CPTP$ map $\Phi(\rho)=\sum_nK_n\rho K_n^\dag$;

$(B4)$ $\sum_ip_iC(\rho_i)\geq C(\rho)$ for $\rho=\sum_ip_i\rho_i$.

Similar to entanglement, the function $C$ is a coherence monotone if it satisfies condition $(B1)$ and $(B2)$.

\subsection{Accessible coherence and Source coherence}\label{sec:ACAS}

For a given state $\rho$ in a Hilbert space $\mathcal{H}$ with finite dimension $d$, we denote by $M^{\mathcal{O}}_a(\rho)$ the set of states that can be reached from $\rho$ via free operations in the set $\mathcal{O}$, and denote by $M^{\mathcal{O}}_s(\rho)$ the set of states that can reach $\rho$ via free operations in $\mathcal{O}$. 
We define two related magnitudes: the \emph{accessible volume}, $V^{\mathcal{O}}_a(\rho)=\mu(M^{\mathcal{O}}_a(\rho))$, which quantifies the volume of states that can be reached by the state $\rho$, and the \emph{source volume}, $V^{\mathcal{O}}_s(\rho)=\mu(M^{\mathcal{O}}_s(\rho))$, which quantifies the volume of states that can reach $\rho$ via free operations. Here, $\mu$ could be an arbitrary Lebesgue measure which maps the set of density matrices to non-negative real number such that $\mu(M^{\mathcal{O}}_a(\rho))=0$ and $\mu(M^{\mathcal{O}}_s(\rho))$ reaches the maximally source volume if $\rho$ is an incoherent state. 

The operational meaning is clear: If $M^{\mathcal{O}}_a(\rho)$ is  larger than $M^{\mathcal{O}}_a(\rho')$, then the state $\rho$ could potentially be more useful than $\rho'$ in quantum information-processing applications. On the other hand, if $M^{\mathcal{O}}_s(\rho)$ is too small, then not many states is useful than the state $\rho$ for any potential applications, i.e, $\rho$ is very useful than many other states in applying for the resource of coherence. We can then define the accessible coherence and the source coherence as follows:
\begin{eqnarray}
C^{\mathcal{O}}_a(\rho)=\frac{V^{\mathcal{O}}_a(\rho)}{V_a^{\sup,\mathcal{O}}},
\end{eqnarray}
and
\begin{eqnarray}
C^{\mathcal{O}}_s(\rho)=1-\frac{V^{\mathcal{O}}_s(\rho)}{V_s^{\sup,\mathcal{O}}},
\end{eqnarray}
where $V_a^{\sup,\mathcal{O}}$  ($V_s^{\sup,\mathcal{O}}$) denotes the maximal accessible (source) volume according to the measure $\mu$. 

We will now show that both accessible coherence and source coherence are coherence monotones. 
\begin{theorem}\label{coherencemonotones}
Both accessible coherence and source coherence satisfy the conditions $(B1)$ and $(B2)$, thus they are coherence monotones.
\end{theorem}
The proof can be found in Supplemental Material \ref{proof_coherencemonotones}.
  
\begin{figure}
\centering
\includegraphics[width=0.5\textwidth]{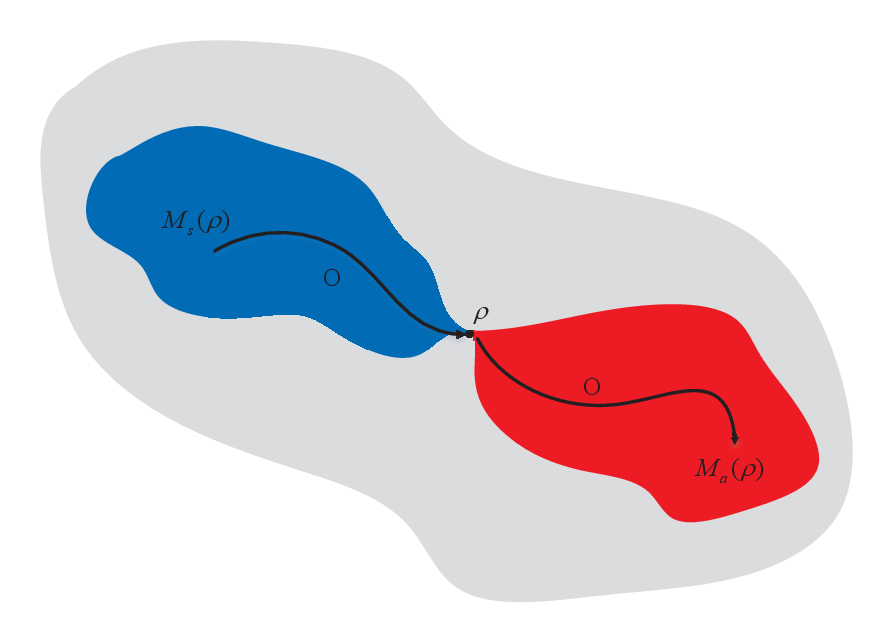}
\caption{(color online). In this schematic figure the source set, $M^\mathcal{O}_s(\rho)$, and the accessible set, $M^\mathcal{O}_a(\rho)$, of the state $\rho$ are depicted. Any state in $M^\mathcal{O}_s(\rho)$ can be transformed to $\rho$ via $\mathcal{O}$ and $\rho$ can be transformed into any state in $M^\mathcal{O}_a(\rho)$ via $\mathcal{O}$.}
\label{figVaVs1}
\end{figure}

\begin{remark} Just as the case in entanglement theory~\cite{Schwaiger15,Sauerwein15}, we would also be interested in the transformations between specific class of coherent states (e.g., the pure states transformation), and the specific volumes $V^\mathcal{O}_a(\rho)$ and $V^\mathcal{O}_s(\rho)$ are only supported on these classes. For example, for single-qubit states, any state can be represented as a point on (or in) the Bloch sphere. We can choose the superficial area on the sphere as the specific volumes $V^\mathcal{O}_a(\rho)$ and $V^\mathcal{O}_s(\rho)$ for this pure states transformation. Meanwhile, If a pure state $\ket{\psi}$ can be transformed to $\rho=\sum_ip_i\ketbra{\psi_i}{\psi_i}$ via free operations $\mathcal{O}$, there will be a channel $\Phi^\mathcal{O}=\sum_ip_i\Phi^\mathcal{O}_i$ corresponding to $\mathcal{O}$, where $\Phi^\mathcal{O}_i(\ket{\psi})=\ket{\psi_i}$ for any $i$. It means that the more pure states we obtained, the more generic states will be obtained. Thus, in this sense, the proficiency of a pure state to generate other pure states characterizes its accessible coherence.
\end{remark}
 
In the following, we will derive explicit formulas for source coherence of pure states transforms via $\mathcal{O}\in\{LSICC,LICC\}$. We consider representatives of LIU classes. To obtain the source coherence, we have:

\begin{theorem}\label{SC_LICC} The source coherence of a bipartite state $\ket{\phi}=\sum_{i=1}^d\sqrt{\lambda_i}\ket{ii}$ with sorted Schmidt vector $\lambda(\phi_{AB})=(\lambda_1,\lambda_2,...,\lambda_d)$ is given by
\begin{eqnarray}\label{eq:Cs_high2}
C_s^{\mathcal{O}}(\ket{\psi})=1-\sum_{\pi\in\sum_d}\frac{[\sum_{k=1}^{d}\pi(k)\lambda_k-\frac{d+1}{2}]^{d-1}}{\Pi_{k=1}^{d-1}\pi(k)-\pi(k+1)},
\end{eqnarray}
where $\mathcal{O}\in\{LSICC,LICC\}$.
\end{theorem}

The proof can be found in Supplemental Material~\ref{sec:LOCC trans}.

Note that, for a state $\ket{\psi}=\sum_i\sqrt{\lambda_i}\ket{i}$ and its ``maximally correlated" state $\ket{\psi'}=\sum_i\sqrt{\lambda_i}\ket{ii}$, $\lambda^\downarrow (\Delta(\psi))=\lambda^\downarrow (\psi')$. This implies that both accessible coherence and source coherence are the same for those two states.

\begin{example}[LICC and LSICC transformations of qitrit-qutrit pure states]
\begin{figure}
\centering
\includegraphics[width=0.5\textwidth]{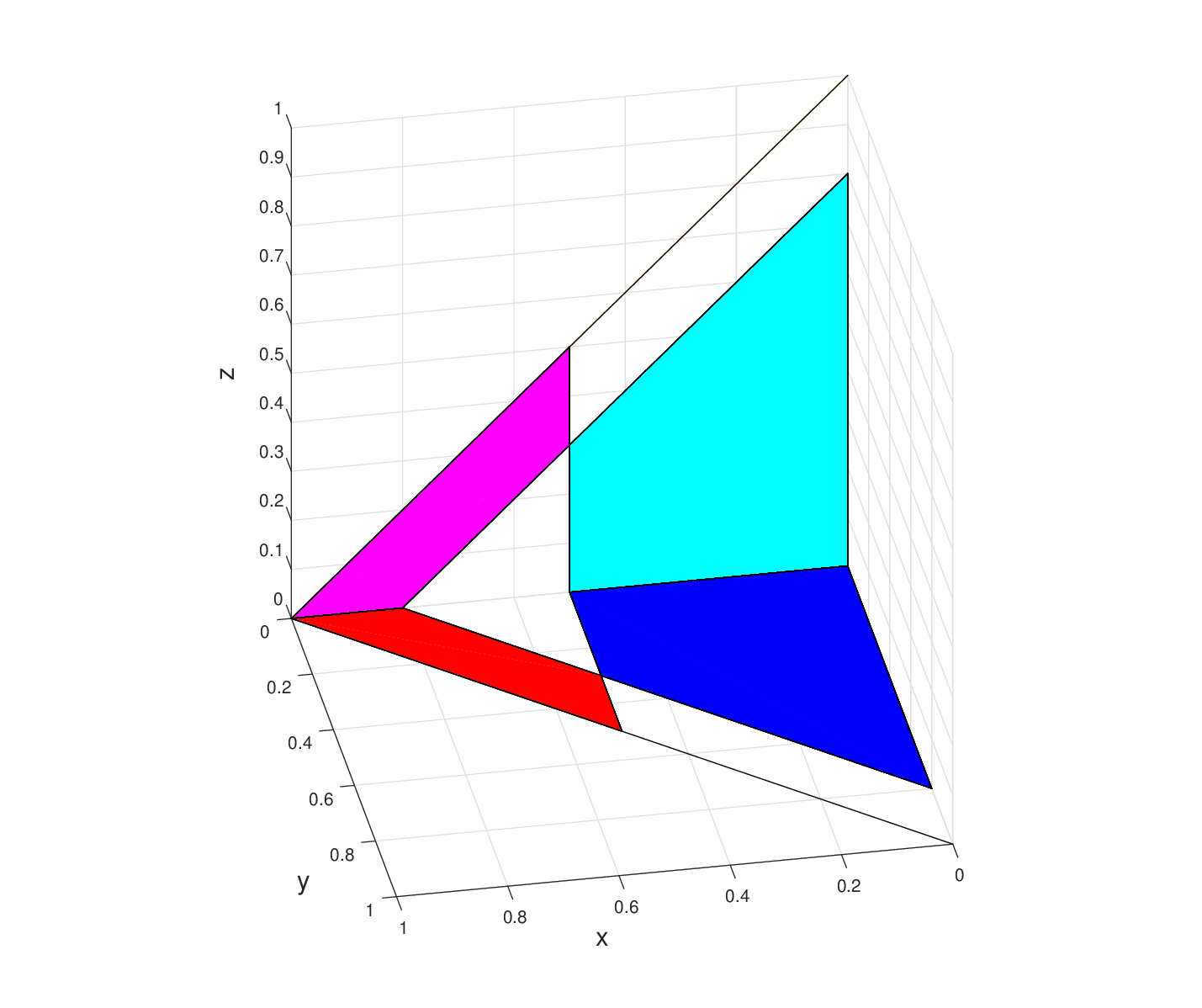}
\caption{(color online). The source set (blue), $M_s(\ket{\psi})$, and the accessible set (red), $M_a(\ket{\psi})$, of the state $\ket{\psi}$ are depicted. Meanwhile, the source set (cyan), $M_s(\ket{\phi})$, and the accessible set (magenta), $M_a(\ket{\phi})$, of the state $\ket{\phi}$ are depicted. Here, the $x$-axis represents the basis vectors $\ket{00}$, the $y$-axis represents the basis vectors $\ket{01}$ and the $z$-axis represents the basis vectors $\ket{11}$. In the figure, $\ket{\psi}=\sqrt{0.5}\ket{00}+\sqrt{0.3}\ket{11}+\sqrt{0.2}\ket{22}$ and $\ket{\phi}=\sqrt{0.5}\ket{00}+\sqrt{0.3}\ket{01}+\sqrt{0.2}\ket{02}$. The source sets (accessible sets) are indeed difference between these two states.}
\label{qutritqutrit}
\end{figure}

Consider the following two pure states $\ket{\psi}=\sqrt a\ket{00}+\sqrt b\ket{11}+\sqrt c\ket{22}$ and $\ket{\phi}=\sqrt a\ket{00}+\sqrt b\ket{01}+\sqrt c\ket{02}$ with $a\geq b\geq c\geq0$ and $c=1-a-b$. Observe that $\ket{\psi}$ can be represented as a point in the $x$-$y$ plane, where the $x$-axis represents the basis $\ket{00}$ and the $y$-axis represents the basis vector $\ket{01}$. Similarly, $\ket{\phi}$ can be represented as a point in the $x$-$z$ plane, where, in addition, the $z$-axis represents the basis vector $\ket{11}$. The accessible volume and the source volume are

\begin{align}
V_{a}^{{\mathcal{O}}}(\ket{\psi})=V_{a}^{{\mathcal{O}}}(\ket{\phi}) = \frac{1}{2}[(1-a)^2-b^2],
\end{align}

\begin{align}
V_{s}^{{\mathcal{O}}}(\ket{\psi})=V_{s}^{{\mathcal{O}}}(\ket{\phi}) = \frac{1}{2}[(a+b)^2-b^2],
\end{align}
where $\mathcal{O}\in\{LSICC,LICC\}$. Thus, the accessible coherence and the source coherence are
\begin{align}
C_{a}^{{\mathcal{O}}}(\ket{\psi})=C_{a}^{{\mathcal{O}}}(\ket{\phi}) =(1-a)^2-b^2,
\end{align}

\begin{align}
C_{s}^{{\mathcal{O}}}(\ket{\psi})=C_{s}^{{\mathcal{O}}}(\ket{\psi})=1-(a+b)^2+b^2.
\end{align}

Figure~\ref{qutritqutrit} shows the accessible set and the source set are difference between $\ket{\psi}$ and $\ket{\phi}$.
\end{example}

\section{Conclusion}\label{sec:conclusion}
In this paper, we studied multipartite coherence structures, and characterized the LICC and SLICC inequivalent classes. In contrast to SLOCC inequivalent classes in entanglement theory, the number of SLICC inequivalent classes of two-qubit pure states in coherence theory is already infinite. We also introduced two coherence monotones: accessible coherence and source coherence, attempting to quantifying the operational quality of multipartite coherence. These two new coherence monotones can be applied in many scenarios (such as PIO, IC, LICC, LSICC). We analyzed pure (or mixed) states via $IC$ and derived explicit formulas for the source coherence. We also showed how the accessible coherence can be computed numerically and gave examples. We hope these operational  monotones will assist with understanding general quantum resource theories.

\begin{acknowledgments}
We thank Eric Chitambar, Zhengjun Xi, Akram Youssry and Kun Fang for valuable discussions. M.-H.~Hsieh was supported by an ARC Future Fellowship under Grant FT140100574 and by US Army Research Office for Basic Scientific Research Grant W911NF-17-1-0401. Y.~Li was supported by the National Natural Science Foundation of China under Grant No.~11671244. Y.~Luo was supported by Research Funds for the Central Universities under Grants No.~2016TS060 and 2016CBY003 and the National Natural Science Foundation of China under Grant No.~61671280.
\end{acknowledgments}

\sloppy
\bibliographystyle{apsrev4-1}
\bibliography{main}
\fussy

%%%%%%%%%%%%%%%%%%%%%%%%%%%%%%%%%%%%%%%%%%%%%%%%%%%%%%%%%%%%%%%%%%%%%%%%%%%%%%%%%%%%%%%%%%%%%%%%%%%%%%%%%%%%%%%%%%%%%%%%%%%%%%%%%%%%%%%%%

\clearpage

\onecolumngrid
\begin{center}
\vspace*{\baselineskip}
{\textbf{\large Supplemental Material: \\[3pt] Inequivalent Multipartite Coherence Classes and New Coherence Monotones}}\\[1pt] \quad \\
\end{center}
%\twocolumngrid

\renewcommand{\theequation}{S\arabic{equation}}
\setcounter{equation}{0}
\setcounter{figure}{0}
\setcounter{table}{0}
\setcounter{section}{0}
\setcounter{page}{1}
\makeatletter

%%%%%%%%%%%%%%%%%%%%%%%%%%%%%%%%%%%%%%%%%%%%
%%%%%%%%%%%%%%%%%%%%%%%%%%%%%%%%%%%%%%%%%%%%
%%%%%%%%%%%%%%%%%%%%%%%%%%%%%%%%%%%%%%%%%%%%

\section{Notation and Preliminaries} \label{sec:Preliminary}
We first introduce the necessary notation.  We consider a Hilbert space $\mathcal{H}$ of finite dimension $d$. The incoherent basis of $\mathcal{H}$ is fixed and is denoted as $\{\ket{i}\}_{i=1}^d$ throughout this paper. A unitary operation $U$ is called an incoherent unitary (\textbf{IU}) if $U=\sum_{i=1}^de^{i\theta_{i}}\ket{i}\bra{\pi(i)}$ with $\pi(i)$ being a permutation. Given a quantum state $\rho$, its von Neumann entropy is $S(\rho)=-\Tr{\rho\log\rho}$.  For an $N$-partite state $\ket{\psi}$ defined on $\mathcal{H}_1\otimes\mathcal{H}_2\otimes\cdots\otimes\mathcal{H}_N$, its reduced density operator on a subset $X\subset [N]:=\{1,2,\cdots,N\}$ is denoted as $\rho^{\psi}_X = \text{Tr}_{\overline{X}} \ket{\psi}\bra{\psi} $, where $\overline{X}=[N]\backslash X$.  

A general resource theory for a quantum system has two components: free states and free operations. In the resource theory of coherence, a free state $\sigma$ (incoherent state)  can be written as $\sigma=\sum_i\sigma_i\ketbra{i}{i}$ for a fixed basis $\{\ket{i}\}$. Variants of the free operations in the resource theory of coherence have been proposed. A completely positive and trace-preserving (CPTP) map $\Phi$ is said to be incoherent operations (\textbf{IC}), if its Kraus operators $K_n$ are of the form $K_n=\sum_i c(i)\ket{j(i)}\bra{i}$ with $\ket{j(i)}$ being a (possibly many-to-one) function from the index set of the basis onto itself, and coefficients $c(i)$ satisfying $\sum_{n}K_n^{\dag}K_n=I$~\cite{Baumgratz14}. If every $K_n=\sum_i c(i)\ket{i}\bra{\pi(i)}$, where $\pi(i)$ is a permutation, then the corresponding operation is a strictly incoherent operation (\textbf{SIO})~\cite{Winter16}.  Lastly, the free operation is called a physically incoherent operation (\textbf{PIO}), if $K_n$ has the form $K_n=U_nP_n$, where $\{U_n\}$ are IU operators and $\{P_n\}$ form an orthogonal and complete set of incoherent projectors~\cite{Chitambar16_PIO}.  From their definitions, we have following inclusion: $PIO\subset SIO\subset IC$.

A fundamental class of operations in entanglement theory is Local Quantum Operations and Classical Communication (\textbf{LOCC}), since it allows an operational definition of ``entanglement'' \cite{Nielsen00}. For two bipartite entangled states $\ket{\phi}$ and $\ket{\psi}$, what is the necessary and sufficient condition for transforming $\ket{\phi}$ to $\ket{\psi}$ using LOCC operations?  Nielsen showed that this condition of entanglement transformation is related to the algebraic theory of $majorization$:

\begin{lemma}[\cite{Nielsen98}]Given two bipartite pure states $\ket{\phi}$ and $\ket{\psi}$ on the system $\mathcal{H}_A\otimes\mathcal{H}_B$, $\ket{\phi}$ can be converted into $\ket{\psi}$ via LOCC if and only if $\lambda(\phi)\prec\lambda(\psi)$, where $\lambda(\phi)$ denotes the vector of eigenvalues of $\Tr_B(\ketbra{\phi}{\phi})$. Here, for two $d$-dimensional vectors $x=(x_1,x_2,...,x_d)$ and $y=(y_1,y_2,...,y_d)$, $x\prec y$ holds if and only if for each $k$ in the range $1,2,...,d$, $\sum_{i=1}^kx^\downarrow_i\leq\sum_{i=1}^ky^\downarrow_i$ with equality when $k=d$, where $x^\downarrow_i$ means that the elements $x_i$ are arranged in decreasing order. 
\end{lemma}

Coherence state transformation has also been studied in the literature, motivated by entanglement state transformation. In the single party setting, Du $et.$ $al.$ obtained the following necessary and sufficient condition of pure state coherence  transformation via $\mathcal{O}\in\{SIO,IC\}$:

\begin{lemma}[\cite{Du15}] \label{Du15}
For a fixed basis $\{\ket{i}\}$, a pure state $\ket{\phi}$ can be converted into $\ket{\psi}$ via $\mathcal{O}\in\{SIO,IC\}$ if and only if $\lambda(\Delta(\phi))\prec\lambda(\Delta(\psi))$, where $\Delta(\rho)=\sum_i\bra{i}\rho\ket{i}\ketbra{i}{i}$, and for convenience, we denote $\Delta(\ketbra{\psi}{\psi})$ as $\Delta(\psi)$. 
\end{lemma}

In the multipartite setting, the class of Local Incoherent Operations and Classical Communication (\textbf{LICC}) can be defined accordingly when the local incoherent operations are IC operators~\cite{Chitambar16, Streltsov_prx}. If the local operations are SIO operations, we call such protocols  \textbf{LSICC}. It is easy to see the following inclusion: $LSICC\subset LICC.$

Chitambar and Hsieh \cite{Chitambar16} studied pure state transformation under LICC and found following:

\begin{lemma}[\cite{Chitambar16}] \label{Chitambar16} Suppose that two bipartite pure states $\ket{\psi}$ and $\ket{\phi}$ have reduced density matrices that are diagonal in the incoherent bases for both parties and both states. Then $\ket{\phi}$ can be converted into $\ket{\psi}$ via LICC if and only if the
squared Schmidt coefficients of $\ket{\psi}$ majorize those of $\ket{\phi}$, i.e., $\ket{\phi}\prec\ket{\psi}$.
\end{lemma}

Shi $et$ $al$.~\cite{Shi1705} and Streltsov $et$ $al$.~\cite{Streltsov1705} studied mixed state transformation of single-qubit systems via $\mathcal{O}\in\{SIO,IC\}$ and obtained the following result.

\begin{lemma}[\cite{Shi1705,Streltsov1705}] The state $\rho= \frac{1}{2}
\begin{pmatrix} 1+r_z & r_x+ir_y   \\ r_x-ir_y & 1-r_z 
\end{pmatrix}$ can be coverted into $\sigma= \frac{1}{2}
\begin{pmatrix} 1+s_z & s_x+is_y   \\ s_x-is_y & 1-s_z 
 \end{pmatrix}$ via SIO, IC if and only if the following inequalities are satisfied: 
 
\begin{eqnarray}
\begin{aligned}
&s_x^2+s_y^2 \leq r_x^2+r_y^2,\\
&\frac{1-r_z^2}{r_x^2+r_y^2}(s_x^2+s_y^2)+s_z^2 \leq1.
\end{aligned}
\end{eqnarray} 

\end{lemma}

Lastly, we also consider LICC protocols that succeed in coherence states transformation only stochastically. Analogous to stochastic LOCC (\textbf{SLOCC}) in entanglement theory, we call these operations stochastic LICC operations and use the notation: \textbf{SLICC}. 

%%%%%%%%%%%%%%%%%%%%%%%%%%%%%%%%%%%%%%%%%%%%
%%%%%%%%%%%%%%%%%%%%%%%%%%%%%%%%%%%%%%%%%%%%
%%%%%%%%%%%%%%%%%%%%%%%%%%%%%%%%%%%%%%%%%%%%

\section{Proof of Lemma \ref{lemma_corank}}\label{sec:proof_corank}
\begin{proof}
Without loss of generality, consider a bipartite coherence pure state $\ket{\psi}=\sum_{i,t=1}^{N,M}\psi_{it}\ket{i}\otimes\ket{t}$ with the number of non-zero product terms in the fixed basis being $NM$. Then for any LICC (resp.~SLICC) protocol with outcome $k$ in the first round, the resulting state is $\ket{\psi'}\equiv F^{(k)}\ket{\psi}=\sum_{i,t=1}^{N,M}c_i^{(k)}\psi_{it}\ket{j(i)}\otimes\ket{t}$, where $F^{(k)}=\sum_{i} c^{(k)}_i \ket{j(i)}\bra{i}$. It is clear that $\ket{\psi'}$ can be expressed as a sum of product terms with no more than $MN$ terms. Consequently, the number of non-zero product terms will not increase as the LICC (resp.~SLICC) protocol continues. 
\end{proof}

%%%%%%%%%%%%%%%%%%%%%%%%%%%%%%%%%%%%%%%%%%%%
%%%%%%%%%%%%%%%%%%%%%%%%%%%%%%%%%%%%%%%%%%%%
%%%%%%%%%%%%%%%%%%%%%%%%%%%%%%%%%%%%%%%%%%%%

\section{Proof of Theorem \ref{The_LICCeq}}\label{sec:proof_LICCeq}
\begin{proof}
It is clear that if $\ket{\psi}$ and $\ket{\phi}$ are LIU equivalent, they are also LICC equivalent since LIU operations performed by each party constitute a special case of LICC protocols.

Next, suppose that two $N$-partite states $\ket{\psi}$ and $\ket{\phi}$ are LICC equivalent, i.e., there exists an LICC protocol that converts $\ket{\psi}$ to $\ket{\phi}$, which consists of many rounds of local IC operations and communications between the parties. Suppose, without loss of generality, that Alice ($A_1$) performs the first local IC operation, yielding the ensemble $\mathcal{E}=\{p_k,\ket{\psi_k}\}$. Since the reduced state of the remaining parties is not changed by Alice's operation, we have
\begin{equation}
\rho^{\psi}_{A_2...A_N}=\sum_kp_k\Tr_{A_1}({\ketbra{\psi_k}{\psi_k}}),
\end{equation} 
where $\rho^{\psi}_{A_2...A_N} = \Tr_{A_1}(\ket{\psi}\bra{\psi})$. Note that, the average entropy is unchanged~\cite{Bennett00}: $S(\rho^{\psi}_{A_1})=S(\rho^{\psi}_{A_2...A_N})=\sum_kp_kS(\Tr_{A_1}({\ketbra{\psi_k}{\psi_k}}))$. From the strict concavity of the von Neumann entropy~\cite{Nielsen00}, it must hold that $S(\Tr_{A_1}({\ketbra{\psi_k}{\psi_k}}))=S(\Tr_{A_1}({\ketbra{\psi}{\psi}}))$, for every $k$, i.e., 
\begin{eqnarray}
\ket{\psi_k}=U_k^{A_1}\otimes I_{A_2...A_N}\ket{\psi},
\end{eqnarray}
where $U^{A_1}_k$ is a unitary operation acting on Alice's system. Since, by assumption, Alice has to perform incoherent operations, then $U_k^{A_1}$ must be local incoherent unitary for all $k$.

Alice can choose a label $k$ with probability $p_k$ as her ``measurement result" and performs the deterministic LIU operation $U^{A_1}_k$ on the state $\ket{\psi}$. Other parties' operations follow similarly so that deterministic LIU operations of the LICC protocol are sufficient to obtain the final the state $\ket{\phi}$. Hence, we conclude that if $\ket{\phi}$ and $\ket{\psi}$ are LICC equivalent, they are also LIU equivalent. 
\end{proof}

%%%%%%%%%%%%%%%%%%%%%%%%%%%%%%%%%%%%%%%%%%%%
%%%%%%%%%%%%%%%%%%%%%%%%%%%%%%%%%%%%%%%%%%%%
%%%%%%%%%%%%%%%%%%%%%%%%%%%%%%%%%%%%%%%%%%%%

\section{Proof of Theorem \ref{The_LSICC}}\label{sec:proof_SLICC}
First, we have following lemma:
\begin{lemma}
If $E_0$ is an IC measurement, then IC measurements $E_i$ can be constructed, such that $\sum_iE_i=I$.
\end{lemma}

\begin{proof}
Let $E_1=I-E_0$. Since $E_1\geq0$, we have spectral decomposition for $E_1=\sum_i\lambda_i\ketbra{\psi_i}{\psi_i}$. Setting $M_i=\sqrt{\lambda_i}\ketbra{i}{\psi_i}$, we find that $E_i=M_i^{\dag}M_i$ is IC POVM for every $i$. 
\end{proof}

Now, we give the proof of Theorem \ref{The_LSICC}.
\begin{proof} 
If $\ket{\phi}=A_1\otimes A_2\otimes \cdots \otimes A_N\ket{\psi}$ holds with SIO (invertible IC) operators $A_k$ with $k=1,2,...,N$, then we can find an SLICC protocol for the parties to transform $\ket{\psi}$ into $\ket{\phi}$ with a positive probability of success. Indeed, each party $k$ can perform an M-outcome IC measurement $\{F_0^{(k)},F_1^{(k)},\ldots,F_M^{(k)}\}$, where $F_0^{(k)}= \sqrt{\frac{p_k}{\bra{\psi_k}A_k^{\dag}A_k\ket{\psi_k}}}A_k$ with $0<  p_k\leq1$. 
It is easy to check that after all parties have performed their corresponding  measurements, the transformation from $\ket{\psi}$ to $\ket{\phi}$ will succeed with probability $p_1p_2\cdots p_N$. The analysis also holds for $\ket{\phi}$ converting into $\ket{\psi}$ by observing that $\ket{\psi}=A^{-1}_1\otimes A^{-1}_2\otimes \cdots \otimes A^{-1}_N\ket{\phi}$.

Conversely, suppose that there is an SLICC protocol, consisting of IC measurements $F^{(k)}$ performed by the $k$-th party, such that $\ket{\psi}$ is transformed into $\ket{\phi}$. Then there must exist one branch of all possible protocol outcomes, say $(x_1,x_2,\cdots,x_N)$, in which $\ket{\phi}$ is obtained. Tracking the performed measurement of each party, $F_{x_k}^{(k)}$, the corresponding IC operators $A_k$ are obtained as follows: 
\begin{eqnarray}\label{eq:The22}
\begin{aligned}
&\frac{1}{\sqrt{p_k}}I_{A_1\cdots A_{k-1}A_{k+1}\cdots A_N}\otimes F_{x_k}^{(k)}\ket{\psi^{(k-1)}} \\
&=I_{A_1\cdots A_{k-1}A_{k+1}\cdots A_N}\otimes A_k\ket{\psi^{(k-1)}}  \\
&=\ket{\psi^{(k)}},
\end{aligned}
\end{eqnarray}
with $p_k=\bra{\psi^{(k-1)}}F_{x_k}^{(k)\dag}F_{x_k}^{(k)}\ket{\psi^{(k-1)}}$, $\ket{\psi^{(0)}}=\ket{\psi}$, and $\ket{\psi^{(N)}}=\ket{\phi}$. In summary,
\begin{eqnarray}\label{eq:The23}
\ket{\phi}=A_1\otimes A_2\otimes \cdots \otimes A_N\ket{\psi}.
\end{eqnarray}
Lastly, $A_k$ must be full rank (hence revertible) since the number of non-zero product terms in $\ket{\psi}$ and $\ket{\phi}$ must be equal, a consequence of Lemma~\ref{lem:product}. 
\end{proof}

%%%%%%%%%%%%%%%%%%%%%%%%%%%%%%%%%%%%%%%%%%%%
%%%%%%%%%%%%%%%%%%%%%%%%%%%%%%%%%%%%%%%%%%%%
%%%%%%%%%%%%%%%%%%%%%%%%%%%%%%%%%%%%%%%%%%%%

\section{Detailed Analysis of Example [Characterization of two-qubit coherence states]}\label{sec:proof_ex_SLICC2}

From Lemma~\ref{lem:product}, we know that the classification is restricted by the number of product terms under SLICC, thus we list these potential equivalent classes based on the number of product terms. The classification of the number of product terms equal to 1 is trivial: every fixed basis $\ket{ij}$ with $i,j\in\{0,1\}$ can be converted to each other via local SIO operators. 

When the number of product terms equals 2, we can conclude the following three classes: $a\ket{00}+b\ket{01}$, $a\ket{00}+c\ket{10}$, and $a\ket{00}+d\ket{11}$, after considering local SIO operators allowed by each party. A similar method can be used when the number of product terms equals  3.

When the number of product terms equals 4, let $\ket{\psi}=a\ket{00}+b\ket{01}+c\ket{10}+d\ket{11}$. If $\ket{\psi}$ can be converted to $\ket{\phi}=a'\ket{00}+b'\ket{01}+c'\ket{10}+d'\ket{11}$ under SLICC, then there exist local SIO operators $A$ and $B$, such that $\ket{\phi}=A\otimes B\ket{\psi}$. Each of the SIO operators $A$, $B$ can be expressed in two forms: $\begin{pmatrix} x & 0   \\ 0 & y \end{pmatrix}$ or $\begin{pmatrix} 0 & z   \\ w & 0  \end{pmatrix}$. If $A=\begin{pmatrix} x & 0   \\ 0 & y \end{pmatrix}$ and $B=\begin{pmatrix} z & 0   \\ 0 & w \end{pmatrix}$ or $A=\begin{pmatrix} 0 & x   \\ y & 0 \end{pmatrix}$ and $B=\begin{pmatrix} 0 & z   \\ w & 0 \end{pmatrix}$, we find that $\frac{ad}{bc}=\frac{a'd'}{b'c'}$ after  the SIO operations. If $A=\begin{pmatrix} 0 & x   \\ y & 0 \end{pmatrix}$ and $B=\begin{pmatrix} z & 0   \\ 0 & w \end{pmatrix}$ or $A=\begin{pmatrix} x & 0   \\ 0 & y \end{pmatrix}$ and $B=\begin{pmatrix} 0 & z   \\ w & 0 \end{pmatrix}$, we find that $\frac{ad}{bc}=\frac{b'c'}{a'd'}$ after  the SIO operations. Denote by $\mathcal{r}=\frac{ad}{bc}$, we conclude that states with the same number $\mathcal{r}$ or $\frac{1}{\mathcal{r}}$ are in the same equivalent class under such a transformation.

%%%%%%%%%%%%%%%%%%%%%%%%%%%%%%%%%%%%%%%%%%%%
%%%%%%%%%%%%%%%%%%%%%%%%%%%%%%%%%%%%%%%%%%%%
%%%%%%%%%%%%%%%%%%%%%%%%%%%%%%%%%%%%%%%%%%%%

\section{Proof of Corollary \ref{Cor_SLICC}}\label{sec:proof_Cor_SLICC}
\begin{proof}
As discussed above, we find the following 4 operators
\begin{eqnarray}
\begin{pmatrix} \frac{\alpha}{b\beta} & 0   \\ 0 & \frac{1}{d} \end{pmatrix}\otimes\begin{pmatrix} \frac{d\alpha}{c} & 0   \\ 0 & \beta \end{pmatrix},
\end{eqnarray}
\begin{eqnarray}
\begin{pmatrix} \frac{\alpha}{b\beta} & 0   \\ 0 & \frac{1}{c} \end{pmatrix}\otimes\begin{pmatrix} 0 & \frac{c\alpha}{d}   \\ \beta & 0 \end{pmatrix},
\end{eqnarray}
\begin{eqnarray}
\begin{pmatrix} 0 & \frac{\alpha}{b\beta}   \\ \frac{1}{b} & 0 \end{pmatrix}\otimes\begin{pmatrix} \frac{b\alpha}{a} & 0   \\ 0 & \beta \end{pmatrix},
\end{eqnarray}
\begin{eqnarray}
\begin{pmatrix} 0 & \frac{\alpha}{c\beta}   \\ \frac{1}{a} & 0 \end{pmatrix}\otimes\begin{pmatrix} 0 & \frac{a\alpha}{b}   \\ \beta & 0 \end{pmatrix},
\end{eqnarray}
can transform $\ket{\psi}=a\ket{00}+b\ket{01}+c\ket{10}+d\ket{11}$ to $\ket{\psi'}=\alpha(\ket{00}+\ket{01}+\ket{10})+\beta\ket{11}$, where $0<|\beta|<1$ and $\alpha=\frac{\sqrt{1-|\beta|^2}}{\sqrt3}>0$. Thus, the state $\ket{\psi}=a\ket{00}+b\ket{01}+c\ket{10}+d\ket{11}$ is SLICC equivalent to the state $\ket{\psi'}=\alpha(\ket{00}+\ket{01}+\ket{10})+\beta\ket{11}$ via those operators above. 
\end{proof}

%%%%%%%%%%%%%%%%%%%%%%%%%%%%%%%%%%%%%%%%%%%%
%%%%%%%%%%%%%%%%%%%%%%%%%%%%%%%%%%%%%%%%%%%%
%%%%%%%%%%%%%%%%%%%%%%%%%%%%%%%%%%%%%%%%%%%%

\section{Proof of Theorem \ref{coherencemonotones}}\label{proof_coherencemonotones}

\begin{proof}
It is evident that if a state $\rho$ is an incoherent state, then $\mu(M^{\mathcal{O}}_a(\rho))=0$ and $\mu(M^{\mathcal{O}}_s(\rho))$ reaches the maximal source volume, i.e., $C^{\mathcal{O}}_a(\rho)=C^{\mathcal{O}}_s(\rho)=0$ if $\rho\in\mathcal{I}$. We also have $C^{\mathcal{O}}_a(\rho)\geq0$ and $C^{\mathcal{O}}_s(\rho)\geq0$ for all quantum states. Condition $(B1)$ therefore holds.

Let $\sigma_\rho$ be the state that can be reached from $\rho$ via free operations $\mathcal{O}$. By definition, $M^{\mathcal{O}}_a(\sigma_\rho)\subset M^{\mathcal{O}}_a(\rho)$, which means that $C^{\mathcal{O}}_a(\rho)\geq C^{\mathcal{O}}_a(\sigma_\rho)$. On the other hand, let $\rho$ be the state that can be reached from $\rho_\sigma$ via free operations $\mathcal{O}$. Then, $M^{\mathcal{O}}_s(\rho)\subset M^{\mathcal{O}}_s(\rho_\sigma)$, which means that $C^{\mathcal{O}}_s(\rho_\sigma)\geq C^{\mathcal{O}}_s(\rho)$. Condition $(B2)$ therefore holds. For Conditions $(B3)$ and $(B4)$, we can construct examples to show they do not hold (see Supplemental Material \ref{sec:antiexample}).
\end{proof}

%%%%%%%%%%%%%%%%%%%%%%%%%%%%%%%%%%%%%%%%%%%%
%%%%%%%%%%%%%%%%%%%%%%%%%%%%%%%%%%%%%%%%%%%%
%%%%%%%%%%%%%%%%%%%%%%%%%%%%%%%%%%%%%%%%%%%%

\section{$C_a^{\mathcal{O}}(\rho)$ and $C_s^{\mathcal{O}}(\rho)$ are coherence monotones that do not satisfy condition $(B3)$ and $(B4)$}
\label{sec:antiexample}
In this section, we will first discuss the accessible coherence $C_a^{\mathcal{O}}(\rho)$. For condition $(B1)$, we can see that if $\rho\in\mathcal{I}$, then $C_a^{\mathcal{O}}(\rho)=0$. For condition $(B2)$, it is easy to see that $M_a^{\mathcal{O}}(\Phi(\rho))\subset M_a^{\mathcal{O}}(\rho)$ for any incoherent operations $\Phi$, thus $C_a^{\mathcal{O}}(\rho)\geq C_a^{\mathcal{O}}(\Phi(\rho))$. 

Now, we will show that $C_a^{\mathcal{O}}(\rho)$ that do not satisfy condition $(B4)$. A single-qubit state $\rho$ can be represented as $\rho=\frac{1}{2}\begin{pmatrix} 1+z&te^{-i\theta}\\ te^{i\theta}&1-z \end{pmatrix}$, where $-1\leq z\leq1$, $0\leq t\leq1$, and $0\leq\theta\leq\pi$. Since accessible coherence $C_a^{\mathcal{O}}(\rho)$ satisfies condition $(B2)$, without loss of generality, we can always consider the state as $\rho=\frac{1}{2}\begin{pmatrix} 1+z&t\\ t&1-z \end{pmatrix}$, where $t^2+z^2\leq1$. We can see $\rho=\lambda_1\ket{\lambda_1}\bra{\lambda_1}+\lambda_2\ket{\lambda_2}\bra{\lambda_2}$, where $\lambda_1=\frac{1+\sqrt{t^2+z^2}}{2}$, $\lambda_2=\frac{1-\sqrt{t^2+z^2}}{2}$ and

$$\ket{\lambda_1}=\frac{t}{2}\frac{z+\sqrt{t^2+z^2}}{t^2+z^2+z\sqrt{t^2+z^2}}\ket{0}+\frac{t}{2}\frac{t}{t^2+z^2+z\sqrt{t^2+z^2}}\ket{1},$$ 

$$\ket{\lambda_2}=\frac{t}{2}\frac{z-\sqrt{t^2+z^2}}{t^2+z^2-z\sqrt{t^2+z^2}}\ket{0}+\frac{t}{2}\frac{t}{t^2+z^2-z\sqrt{t^2+z^2}}\ket{1}.$$ 

Let $t=z=0.1$, we can find $C_a^{\mathcal{O}}(\rho)-\lambda_1 C_a^{\mathcal{O}}(\ket{\lambda_1})-\lambda_2 C_a^{\mathcal{O}}(\ket{\lambda_2})=0.0994$. Thus, $C_a^{\mathcal{O}}$ is not convex.

For condition $(B3)$, consider a general amplitude damping channel~\cite{Nielsen00} with $E_0=\sqrt{p}\begin{pmatrix} 1&0\\ 0&\sqrt{1-\gamma} \end{pmatrix}$, $E_1=\sqrt{p}\begin{pmatrix} 0&\sqrt{\gamma}\\ 0&0 \end{pmatrix}$, $E_2=\sqrt{1-p}\begin{pmatrix} \sqrt{1-\gamma}&0\\ 0&1 \end{pmatrix}$ and $E_3=\sqrt{1-p}\begin{pmatrix} 0&0\\ \sqrt{\gamma}&0 \end{pmatrix}$. Let $p=0.99$ and $\gamma=t=z=0.5$, we can find $C_a^{\mathcal{O}}(\rho)- \sum_nC_a^{\mathcal{O}}(\rho_n)=-0.1912$, thus condition $(B3)$ does not hold.

We will now show that $C_s^{\mathcal{O}}(\rho)$ does not satisfy condition $(B4)$. Let $t=z=0.1$, we can find $C_s^{\mathcal{O}}(\rho)-\lambda_1 C_s^{\mathcal{O}}(\ket{\lambda_1})-\lambda_2 C_s^{\mathcal{O}}(\ket{\lambda_2})=0.6930$. Thus, $C_s^{\mathcal{O}}$ is not convex. For condition $(B3)$, consider a general amplitude damping channel~\cite{Nielsen00} with $E_0=\sqrt{p}\begin{pmatrix} 1&0\\ 0&\sqrt{1-\gamma} \end{pmatrix}$, $E_1=\sqrt{p}\begin{pmatrix} 0&\sqrt{\gamma}\\ 0&0 \end{pmatrix}$, $E_2=\sqrt{1-p}\begin{pmatrix} \sqrt{1-\gamma}&0\\ 0&1 \end{pmatrix}$ and $E_3=\sqrt{1-p}\begin{pmatrix} 0&0\\ \sqrt{\gamma}&0 \end{pmatrix}$. Let $p=0.99$ and $\gamma=0.8$, $t=z=0.4$, we find that $C_s^{\mathcal{O}}(\rho)- \sum_nC_s^{\mathcal{O}}(\rho_n)=-0.2123$, thus the condition $(B3)$ does not hold.

%%%%%%%%%%%%%%%%%%%%%%%%%%%%%%%%%%%%%%%%%%%%%%%%%%%%%%%%%%%%%%%%%%%%%%%%%%%%%%%%%%%%%%%%%%%%%%%%%%%%%%%%%%%%%%%%%%%%%%%%%%%%%%%%%%%%%%%%%%%%%%%%%%%%%%%
%%%%%%%%%%%%%%%%%%%%%%%%%%%%%%%%%%%%%%%%%%%%%%%%%%%%%%%%%%%%%%%%%%%%%%%%%%%%%%%%%%%%%%%%%%%%%%%%%%%%%%%%%%%%%%%%%%%%%%%%%%%%%%%%%%%%%%%%%%%%%%%%%%%%%%

\section{Proof of Theorem \ref{SC_LICC}}\label{sec:LOCC trans}

In this section, $\mu$ is chosen as a measure on the set of LIU equivalence classes. In the proof of $Lemma$ \ref{Chitambar16}, since the IC operators Alice and Bob used are full rank. Thus, the local IC operators are local SIO operators. Then, we have

%\begin{theorem} Suppose $\ket{\psi}_{AB}$ and $\ket{\phi}_{AB}$ have reduced density matrices that are diagonal in the incoherent bases. i.e., $\ket{\psi}_{AB}=\sum_i\sqrt{\psi_i}\ket{ii}$ and $\ket{\phi}_{AB}=\sum_i\sqrt{\phi_i}\ket{ii}$. $\ket{\phi}_{AB}$ can be converted into $\ket{\psi}_{AB}$ via $\mathcal{O}\in\{LSICC,LICC\}$ if and only if the squared Schmidt coefficients of $\ket{\psi}_{AB}$ majorize those of $\ket{\phi}_{AB}$.
%\end{theorem}

First, we will show the explicit formula of source coherence for pure states transforms via $\mathcal{O}\in\{LSICC,LICC\}$. From $Lemma$~\ref{Du15}, we can see the source set of $\ket{\psi}$ is given by
\begin{align}
M_s^{\mathcal{O}}(\psi) = \{\ket{\phi}\in {\cal H} \ \mbox{s.t.} \ \lambda (\phi) \prec \lambda (\psi)\}, \label{Eq_MsMastates11}
\end{align}

Because any pure states in a LIU equivalence class can be seen as the vector $\lambda(\psi)$, we can associate the set given in Eq. (\ref{Eq_MsMastates11}) the following set of sorted vectors in $\R^d$:
\begin{align}
 \mathcal{M}_s^{\mathcal{O}}(\psi) = \{\lambda^\downarrow \in \R^d \ \mbox{s.t.} \ \lambda^\downarrow \prec \lambda(\psi)\},\label{Eq_MsMa_LICC}
\end{align}
where $d$ denotes the Schmidt number of $\ket{\psi}$.

The set given in Eq.(\ref{Eq_MsMa_LICC}) is a convex polytope, and as shown in~\cite{Sauerwein15}, the source set~Eq.(\ref{Eq_MsMa_LICC}) is a simple polytope~\cite{Ziegler95,Brion97}. The simple polytope of the set $\mathcal{M}_s^{\mathcal{O}}(\psi)$ is the some as the polytope of the source set of entanglement as shown in~\cite{Sauerwein15}, thus the volume of $\mathcal{M}_s^{\mathcal{O}}(\psi)$ is
\begin{eqnarray}\label{eq:Vs1}
V_s^{\mathcal{O}}(\ket{\psi})=\frac{1}{d!}\frac{\sqrt{d}}{(d-1)!}\sum_{\pi\in\Sigma_d}\frac{[\sum_{k=1}^d\pi(k)\lambda_k-\frac{d+1}{2}]^{d-1}}{\Pi_{k=1}^{d-1}\pi(k)-\pi(k+1)},
\end{eqnarray}
where $\pi$ denotes an element of the permutation group $\Sigma_d$ of $d$ elements and $\mathcal{O}\in\{LSICC,LICC\}$. Note that for the incoherent state $\ket{\psi_{incoh}}$, the vector $\lambda(\psi_{incoh})=(1,0,0,...,0)$ can be obtained from any other states via $\mathcal{O}\in\{LSICC,LICC\}$, and therefore its source volume is the maximum, i.e., $V_s^{\mathcal{O}}(\ket{\psi_{incoh}})=\sup_{\phi}V_s^{\mathcal{O}}(\ket{\phi})$. The volume is $V_s^{\mathcal{O}}(\ket{\psi_{incoh}})=\frac{\sqrt{d}}{d!(d-1)!}$. For a maximal correlated state $\ket{\psi^+}$, the corresponding vector $\lambda(\psi^{+})=\frac{1}{\sqrt d}(1,1,1,...,1)$. It is straightforward to see that the volume is $V_s^{\mathcal{O}}(\ket{\psi^{+}})=0$.  

Thus, the source coherence of a pure state $\ket{\phi}$ with sorted vector $\lambda(\phi)$ is given by
\begin{eqnarray}\label{eq:Cs_high3}
C_s^{\mathcal{O}}(\ket{\psi})=1-\sum_{\pi\in\sum_d}\frac{[\sum_{k=1}^{d}\pi(k)\lambda_k-\frac{d+1}{2}]^{d-1}}{\Pi_{k=1}^{d-1}\pi(k)-\pi(k+1)},
\end{eqnarray}
where $\mathcal{O}\in\{LSICC,LICC\}$. 

%%%%%%%%%%%%%%%%%%%%%%%%%%%%%%%%%%%%%%%%%%%%%%%%%%%%%%%%%%%%%%%%%%%%%%%%%%%%%%%%%%%%%%%%%%%%%%%%%%%%%%%%%%%%%%%%%%%%%%%%%%%%%%%%%%%%%%%%%%%%%%%%%%%%%%%
\section{Source coherence for pure states transforms via IC and SIO}

In this section, we will discuss source coherence for pure states transforms via $\mathcal{O}\in\{SIO,IC\}$. The method of this proof is similar to the proof of Theorem \ref{SC_LICC}. Note that $\mu$ is chosen as a measure on the set of IU equivalence classes. From $Lemma$~\ref{Du15}, the source set of $\ket{\psi}$ is given by
\begin{align}
 M_s^{\mathcal{O}}(\psi) = \{\ket{\phi}\in {\cal H} \ \mbox{s.t.} \ \lambda (\Delta(\phi)) \prec \lambda (\Delta(\psi))\}, \label{Eq_MsMastates1}
\end{align}

Because any pure states in a IU equivalence class can be seen as the vector $\lambda(\Delta(\psi))$, we can associate the set given in Eq. (\ref{Eq_MsMastates1}) the following sets of sorted vectors in $\R^d$:
\begin{align}
 \mathcal{M}_s^{\mathcal{O}}(\psi) = \{\lambda^\downarrow \in \R^d \ \mbox{s.t.} \ \lambda^\downarrow \prec \lambda(\Delta(\psi))\},\label{Eq_MsMa}
\end{align}
where $d$ denotes the rank of $\Delta(\psi)$; these sets are hence supported on states of the same dimensions as $\Delta(\psi)$.

Then, the set given in Eqs. (\ref{Eq_MsMa}) is also a simple polytope~\cite{Sauerwein15}. Thus the volume of $\mathcal{M}_s^{\mathcal{O}}(\psi)$ is
\begin{eqnarray}\label{eq:Vs2}
V_s^{\mathcal{O}}(\ket{\psi})=\frac{1}{d!}\frac{\sqrt{d}}{(d-1)!}\sum_{\pi\in\Sigma_d}\frac{[\sum_{k=1}^d\pi(k)\lambda_k-\frac{d+1}{2}]^{d-1}}{\Pi_{k=1}^{d-1}\pi(k)-\pi(k+1)},
\end{eqnarray}
where $\pi$ denotes an element of the permutation group $\Sigma_d$ of $d$ elements and $\mathcal{O}\in\{SIO,IC\}$. Note that for the incoherent state $\ket{\psi_{incoh}}$, the vector $\lambda(\Delta(\psi_{incoh}))=(1,0,0,...,0)$ can be obtained from any other states via $\mathcal{O}\in\{SIO,IC\}$, and therefore its source volume is the maximum, i.e., $V_s^{\mathcal{O}}(\ket{\psi_{incoh}})=\sup_{\phi}V_s^{\mathcal{O}}(\ket{\phi})$. The volume is $V_s^{\mathcal{O}}(\ket{\psi_{incoh}})=\frac{\sqrt{d}}{d!(d-1)!}$. For a maximally coherent state $\ket{\psi^+}$, the corresponding vector $\lambda(\Delta(\psi^{+}))=\frac{1}{\sqrt d}(1,1,1,...,1)$. It is straightforward to see that the volume is $V_s^{\mathcal{O}}(\ket{\psi^{+}})=0$.  

Thus, we have following result for a pure state with sorted vector $\lambda(\Delta(\phi))\in\C^d$:

\begin{theorem}The source coherence of a pure state $\ket{\phi}$ with sorted vector $\lambda(\Delta(\phi))$ is given by
\begin{eqnarray}\label{eq:Cs_high1}
C_s^{\mathcal{O}}(\ket{\psi})=1-\sum_{\pi\in\sum_d}\frac{[\sum_{k=1}^{d}\pi(k)\lambda_k-\frac{d+1}{2}]^{d-1}}{\Pi_{k=1}^{d-1}\pi(k)-\pi(k+1)},
\end{eqnarray}
where $\mathcal{O}\in\{SIO,IC\}$. 
\end{theorem}

%%%%%%%%%%%%%%%%%%%%%%%%%%%%%%%%%%%%%%%%%%%
%%%%%%%%%%%%%%%%%%%%%%%%%%%%%%%%%%%%%%%%%%%
%%%%%%%%%%%%%%%%%%%%%%%%%%%%%%%%%%%%%%%%%%%
\section{Examples of source and accessible volumes for low-dimensional systems}\label{sec:example}
In this section, we will explicitly calculate the accessible coherence and the source coherence for low-dimensional states. For the single-qubit case, we refer interested readers to the Supplemental Material for the detailed calculation.

\begin{example}[LICC and LSICC transformations of two-qubit pure states]

\begin{figure}
\centering
\includegraphics[width=0.4\textwidth]{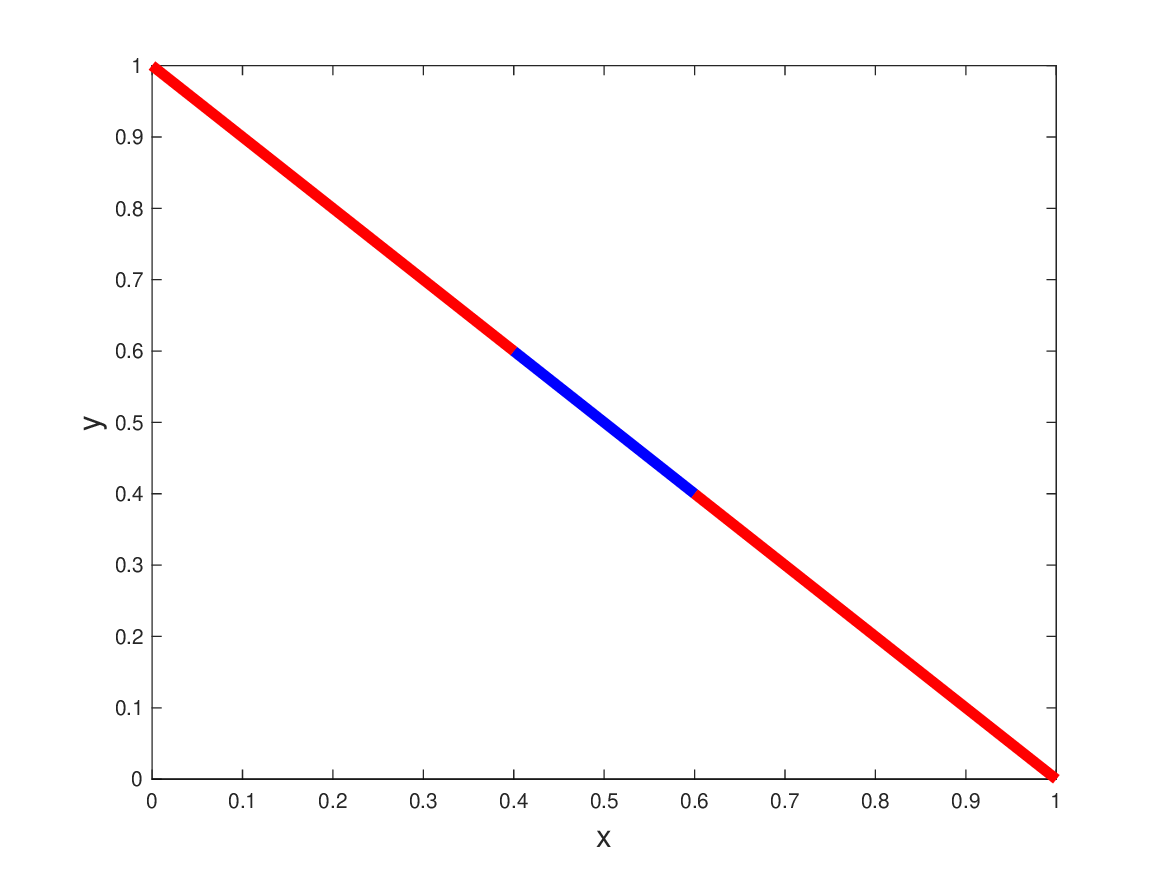}
\caption{(color online). The source set (blue), $M_s(\ket{\psi})$, and the accessible set (red), $M_a(\ket{\psi})$, of the state $\ket{\psi}$ are depicted. Any state in $M_s(\ket{\psi})$ can be transformed to $\ket{\psi}$ via LSICC and LICC and $\ket{\psi}$ can be transformed into any state in $M_a(\ket{\psi})$ via LSICC and LICC. In the figure, $\ket{\psi}=\sqrt{0.6}\ket{01}+\sqrt{0.4}\ket{10}$.}
\label{fig_qubit}
\end{figure}

Let $\ket{\psi}=\sqrt a\ket{01}+\sqrt b\ket{10}$, where $a+b=1$. We consider the state $\ket{\psi'}=\sqrt a\ket{00}+\sqrt b\ket{11}$ with Schmidt coefficients $a$ and $b$. It is then straightforward to find that the accessible volume and source volume are
\begin{align}
V_{a}^{{\mathcal{O}}}(\ket{\psi}) = \sqrt2(x-\frac{1}{2}),
\end{align}

\begin{align}
V_{s}^{{\mathcal{O}}}(\ket{\psi}) = \sqrt2(1-x),
\end{align}
where $x\geq y$ and $x,y\in\{a,b\}$ and $\mathcal{O}\in\{LICC,LSICC\}$. Figure~\ref{fig_qubit} shows this state transformation, thus the accessible coherence and source coherence are 
\begin{align}
C_{a}^{{\mathcal{O}}}(\ket{\psi}) = C_{s}^{{\mathcal{O}}}(\ket{\psi}) =2(1-x).
\end{align}
\end{example}

\subsection{$C_a^{\mathcal{O}}(\rho)$ and $C_s^{\mathcal{O}}(\rho)$ in single-qubit case}\label{sec:single-qubit}
For a single-qubit state, we calculate the coherence monotones by choosing different free operations.

\begin{example}[Accessible Coherence and Source Coherence via PIO]
\begin{figure}
\centering
\includegraphics[width=0.6\textwidth]{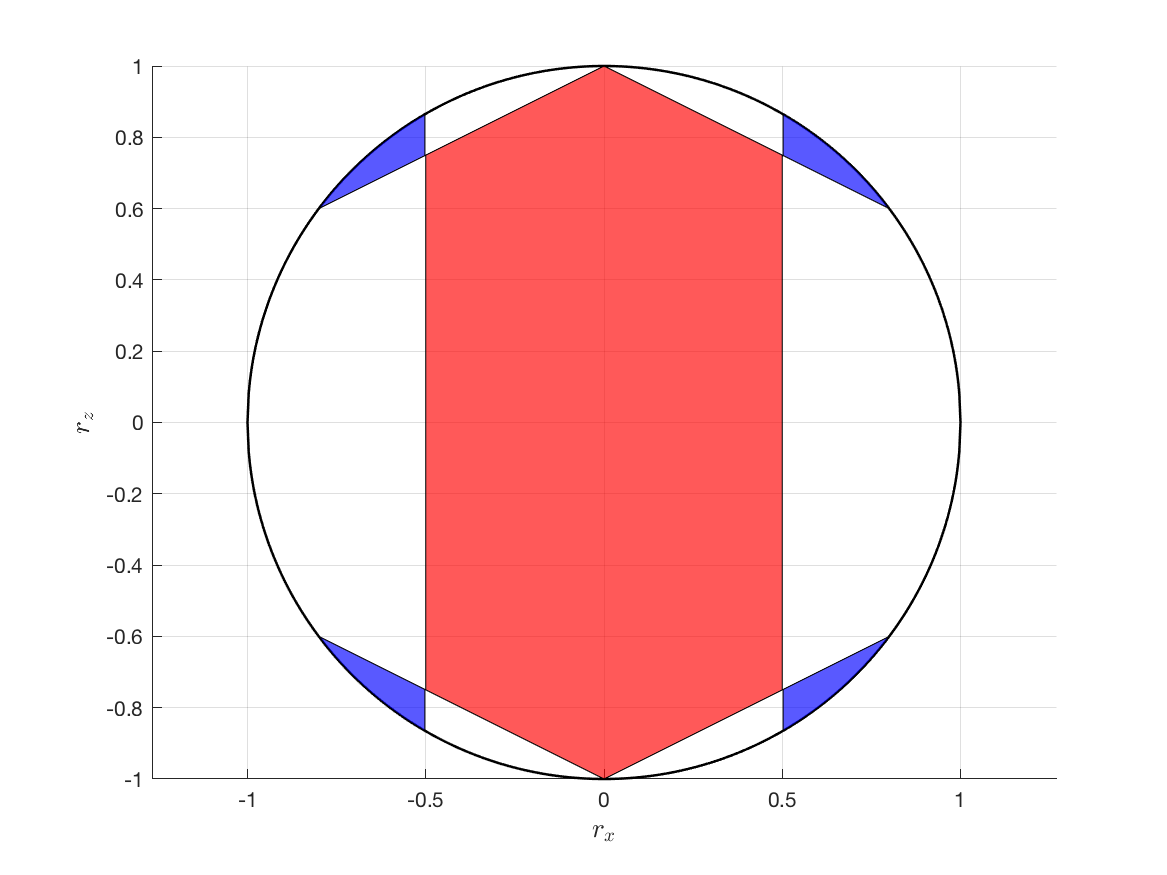}
\caption{(color online). As shown in the figure, the source set (blue corlor), $M_s(\rho)$, and the accessible set (red corlor), $M_a(\rho)$, of the state $\rho$ are depicted. Any state in $M_s(\rho)$ can be transformed to $\rho$ via PIO and $\rho$ can be transformed into any state in $M_a(\rho)$ via PIO. In the figure, the single-qubit $\rho$ has the Bloch vector $(\frac{1}{2},0,\frac{3}{4})$.}
\label{acas_pio}
\end{figure}

As shown in~\cite{Shi1705}, a single-qubit $\rho= \frac{1}{2}
\begin{pmatrix} 1+r_z & r_x+ir_y   \\ r_x-ir_y & 1-r_z 
\end{pmatrix}$ coverted into $\sigma= \frac{1}{2}
\begin{pmatrix} 1+s_z & s_x+is_y   \\ s_x-is_y & 1-s_z 
\end{pmatrix}$ via PIO should satisfy the following equalities:

\begin{eqnarray}
s_x^2+s_y^2\leq r_x^2+r_y^2,
\end{eqnarray}

\begin{eqnarray}
\frac{s_x^2+s_y^2}{r_x^2+r_y^2}\leq\frac{(s_z-1)^2}{(1-r_z)^2},
\end{eqnarray} 

and
\begin{eqnarray}
\frac{s_x^2+s_y^2}{r_x^2+r_y^2}\leq\frac{(s_z+1)^2}{(1-r_z)^2}.
\end{eqnarray} 

Since $\mu$ is chosen as measure on the $IU$ equivalence class, we can consider the transform range on the $x-z$ plane, which is a convex hexagon with six vertexes: $(\pm z,\pm\sqrt{r_x^2+r_y^2})$, and $(\pm1,0)$ (More precisely, the $IU$ equivalence classes of states are considered on the first quadrant, but we consider the $x-z$ plane here, since the symmetry does not affect the ratio $\frac{V^{\mathcal{O}(\rho)}}{V^{\sup}}$). 
 
The volume of accessible set $V_a^{PIO}(\rho)$ is the square of these six vertexes:
\begin{eqnarray}
V_a^{PIO}(\rho)=2|r_z|\sqrt{r_x^2+r_y^2}+2\sqrt{r_x^2+r_y^2}.
\end{eqnarray}

Since $2|r_z|\sqrt{r_x^2+r_y^2}+2\sqrt{r_x^2+r_y^2}\leq r_x^2+r_y^2+r_z^2+2\sqrt{r_x^2+r_y^2}$, the maximal accessible volume can be reached with $r_x^2+r_y^2=r_z^2=\frac{1}{2}$ and $V_a^{\sup,PIO}(\rho)=1+\sqrt2$. Thus, we have:
\begin{eqnarray}
C_a^{PIO}(\rho)=\frac{2|r_z|\sqrt{r_x^2+r_y^2}+2\sqrt{r_x^2+r_y^2}}{1+\sqrt2}.
\end{eqnarray} 
 
On the other hand, a single-qubit $\rho= \frac{1}{2}
\begin{pmatrix} 1+r_z & r_x+ir_y   \\ r_x-ir_y & 1-r_z 
\end{pmatrix}$ can be transformed from $\sigma= \frac{1}{2}
\begin{pmatrix} 1+s_z & s_x+is_y   \\ s_x-is_y & 1-s_z 
\end{pmatrix}$ via PIO should satisfy following equalities:

\begin{eqnarray}
s_x^2+s_y^2\geq r_x^2+r_y^2,
\end{eqnarray}

\begin{eqnarray}
\frac{s_x^2+s_y^2}{r_x^2+r_y^2}\geq\frac{(s_z-1)^2}{(1-r_z)^2},
\end{eqnarray} 

\begin{eqnarray}
\frac{s_x^2+s_y^2}{r_x^2+r_y^2}\geq\frac{(s_z+1)^2}{(1-r_z)^2},
\end{eqnarray} 
 
and
\begin{eqnarray}
s_x^2+s_y^2+s_z^2\leq 1.
\end{eqnarray}
 
The volume of source set $V_a^{PIO}(\rho)$ is:
\begin{equation}
V_s^{PIO}(\rho)=
\begin{cases}
2Q_1+2[\sin{2\arcsin{t}}-\sin{2\arcsin{\sqrt{r_x^2+r_y^2}}}]-S_1
&\text{$\sqrt{r_x^2+r_y^2}+r_z^2\geq1$,} \\
2Q_2-2\sqrt{r_x^2+r_y^2}\sqrt{1-(r_x^2+r_y^2)}-S_2& \text{$\sqrt{r_x^2+r_y^2}+r_z^2\leq1$, $r_x^2+r_y^2\neq0,$}\\
\pi& \text{$r_x^2+r_y^2=0$},
\end{cases}
\end{equation}
where $Q_1=\arcsin{t}-\arcsin{\sqrt{r_x^2+r_y^2}}$, $Q_2=\frac{\pi}{2}-\arcsin{\sqrt{r_x^2+r_y^2}}$, $t=\frac{2\sqrt{r_x^2+r_y^2}(1-|r_z|)}{(r_x^2+r_y^2)+(1-|r_z|)^2}$, $S_1=2(|r_z|+\frac{(r_x^2+r_y^2)-(|r_z|-1)^2}{(r_x^2+r_y^2)+(|r_z|-1)^2})(\frac{2\sqrt{r_x^2+r_y^2}(1-|r_z|)}{(r_x^2+r_y^2)+(1-|r_z|)^2}-\sqrt{r_x^2+r_y^2})$ and $S_2=2\frac{|r_z|}{1-|r_z|}\sqrt{r_x^2+r_y^2}$.

Obviously, when $r_x^2+r_y^2=0$, $V_s^{PIO}(\rho)=V_s^{\sup,{PIO}}=\pi$. Thus we have:
\begin{eqnarray}
C_s^{PIO}(\rho)=
\begin{cases}
\frac{2}{\pi}Q_1+\frac{2}{\pi}[\sin{(2\arcsin{t})}-\sin{(2\arcsin{\sqrt{r_x^2+r_y^2}})}]-\frac{1}{\pi}S_1
&\text{$\sqrt{r_x^2+r_y^2}+r_z^2\geq1$,} \\
\frac{2}{\pi}[\frac{\pi}{2}-\arcsin{\sqrt{r_x^2+r_y^2}}]-\frac{2}{\pi}\sqrt{r_x^2+r_y^2}\sqrt{1-(r_x^2+r_y^2)}-\frac{1}{\pi}S_2& \text{$\sqrt{r_x^2+r_y^2}+r_z^2\leq1$, $r_x^2+r_y^2\neq0$,}\\
1& \text{$r_x^2+r_y^2=0$},
\end{cases}
\end{eqnarray} 
where $Q_1=\arcsin{t}-\arcsin{\sqrt{r_x^2+r_y^2}}$.
\end{example}

\begin{example}[Accessible Coherence and Source Coherence via SIO and IC]
\begin{figure}
\centering
\includegraphics[width=0.6\textwidth]{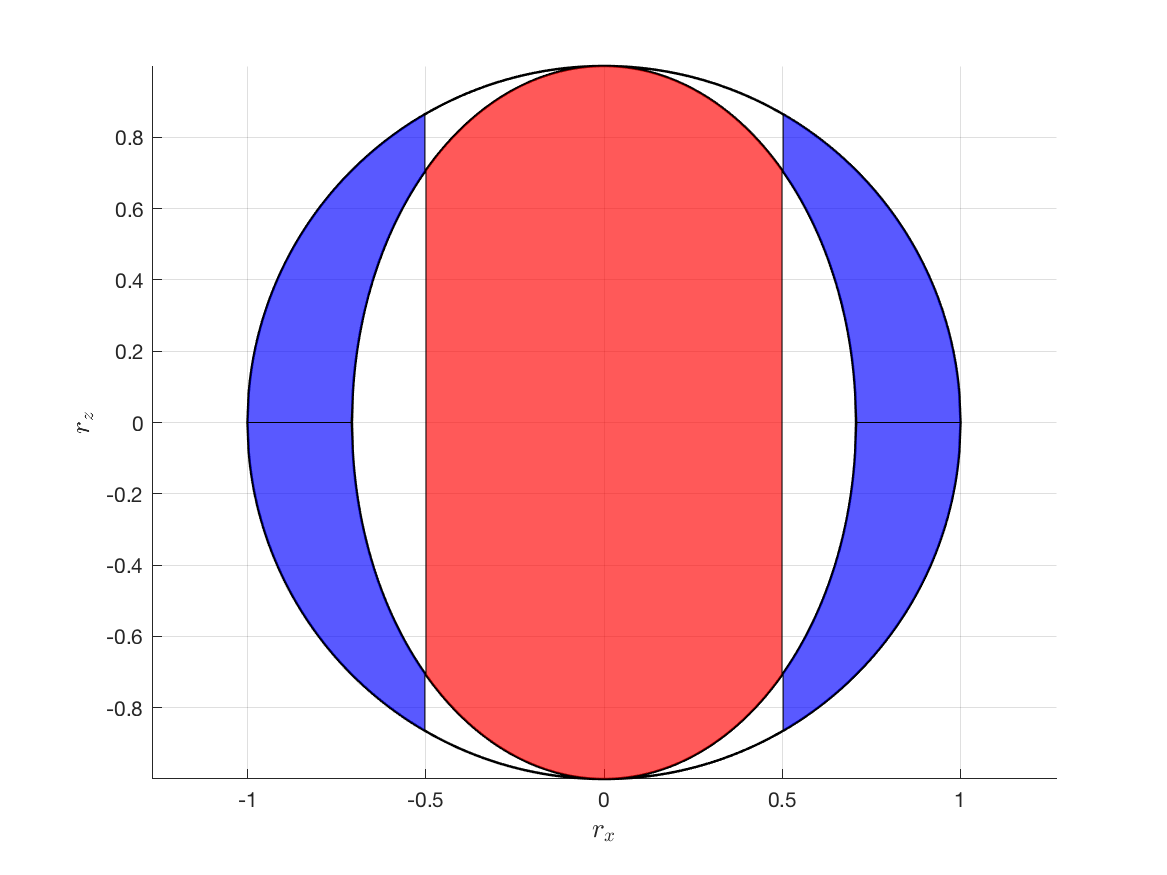}
\caption{(color online). As shown in the figure, the source set (blue corlor), $M_s(\rho)$, and the accessible set (red corlor), $M_a(\rho)$, of the state $\rho$ are depicted. Any state in $M_s(\rho)$ can be transformed to $\rho$ via IC and SIO and $\rho$ can be transformed into any state in $M_a(\rho)$ via IC and SIO. In the figure, the single-qubit $\rho$ has the Bloch vector $(\frac{1}{2},0,\frac{\sqrt2}{2})$.}\label{acas_ic}
\end{figure}

If the set measure $\mu$ denotes the volume of these transform range above, for $\mathcal{O}\in\{SIO,IC\}$, we find that $C_a^{\mathcal{O}}(\rho)$ and $C_s^{\mathcal{O}}(\rho)$ in the qubit case is: 

\begin{eqnarray}
V_a^{\mathcal{O}}(\rho)&=&2\sqrt{\frac{r_x^2+r_y^2}{1-r_z^2}}\arcsin{\sqrt{1-r_z^2}}+2|r_z|\sqrt{r_x^2+r_y^2},
\end{eqnarray}

and
\begin{equation}
V_s^{\mathcal{O}}(\rho)=
\begin{cases}
2\arcsin{\sqrt{1-r_x^2-r_y^2}}-2\sqrt{1-r_x^2-r_y^2}\sqrt{1-r_z^2}-2\sqrt{\frac{r_x^2+r_y^2}{1-r_z^2}}\arcsin{|r_z|}+2|r_z|\sqrt{r_x^2+r_y^2}&\text{$r_x^2+r_y^2+r_z^2\neq1$,} \\
2\arcsin{|r_z|}-2|r_z|\sqrt{1-r_z^2}& \text{$r_x^2+r_y^2+r_z^2=1$}.
\end{cases}
\end{equation}

Note that $\sqrt{r_x^2+r_y^2}$ is the $l_1$-norm coherence of $\rho$. The supremum of $V_a^{\mathcal{O}}(\rho)$ is in the case $z_0=0$, thus $V_a^{\sup,\mathcal{O}}(\rho)=\pi$. Combine $V_a^{\sup}$ with $V_a$, we have

\begin{eqnarray}\label{eq:Ca_mixed}
C_a^{\mathcal{O}}(\rho)=\frac{2}{\pi}(\sqrt{\frac{r_x^2+r_y^2}{1-r_z^2}}\arcsin{\sqrt{1-r_z^2}}+|r_z|\sqrt{r_x^2+r_y^2}).
\end{eqnarray}

If $\rho$ is pure state, we have 
\begin{eqnarray}\label{eq:Ca_pure}
C_a^{\mathcal{O}}(\rho)=\frac{2}{\pi}(\arcsin{\sqrt{1-r_z^2}}+|r_z|\sqrt{1-r_z^2}).
\end{eqnarray} 

We also have $V_s^{\sup}=\pi$, and note that $1-\arcsin{|r_z|}=\arcsin{\sqrt{1-r_z^2}}$, thus
\begin{equation}
C_s^{\mathcal{O}}(\rho)=
\begin{cases}
1-\frac{2}{\pi}(\arcsin{\sqrt{1-r_x^2-r_y^2}}-\sqrt{1-r_x^2-r_y^2}\sqrt{1-r_z^2}-\sqrt{\frac{r_x^2+r_y^2}{1-r_z^2}}\arcsin{|r_z|}+|r_z|\sqrt{r_x^2+r_y^2})&\text{$r_x^2+r_y^2+r_z^2\neq1$,} \\
\frac{2}{\pi}(\arcsin{\sqrt{1-r_z^2}}+|r_z|\sqrt{1-r_z^2})& \text{$r_x^2+r_y^2+r_z^2=1$}.
\end{cases}
\end{equation}

For $SIO$ and $IC$, accessible coherence and source coherence are consistent in the pure state case. 
\end{example}

\subsection{$C_a^{\mathcal{O}}(\rho)$ and $C_s^{\mathcal{O}}(\rho)$ in two-qubit pure state transformation via SIO and IC}\label{sec:two-qubit IC}

\begin{example}[two-qubit pure state transformation via SIO and IC]

\begin{figure}
\centering
\includegraphics[width=0.45\textwidth]{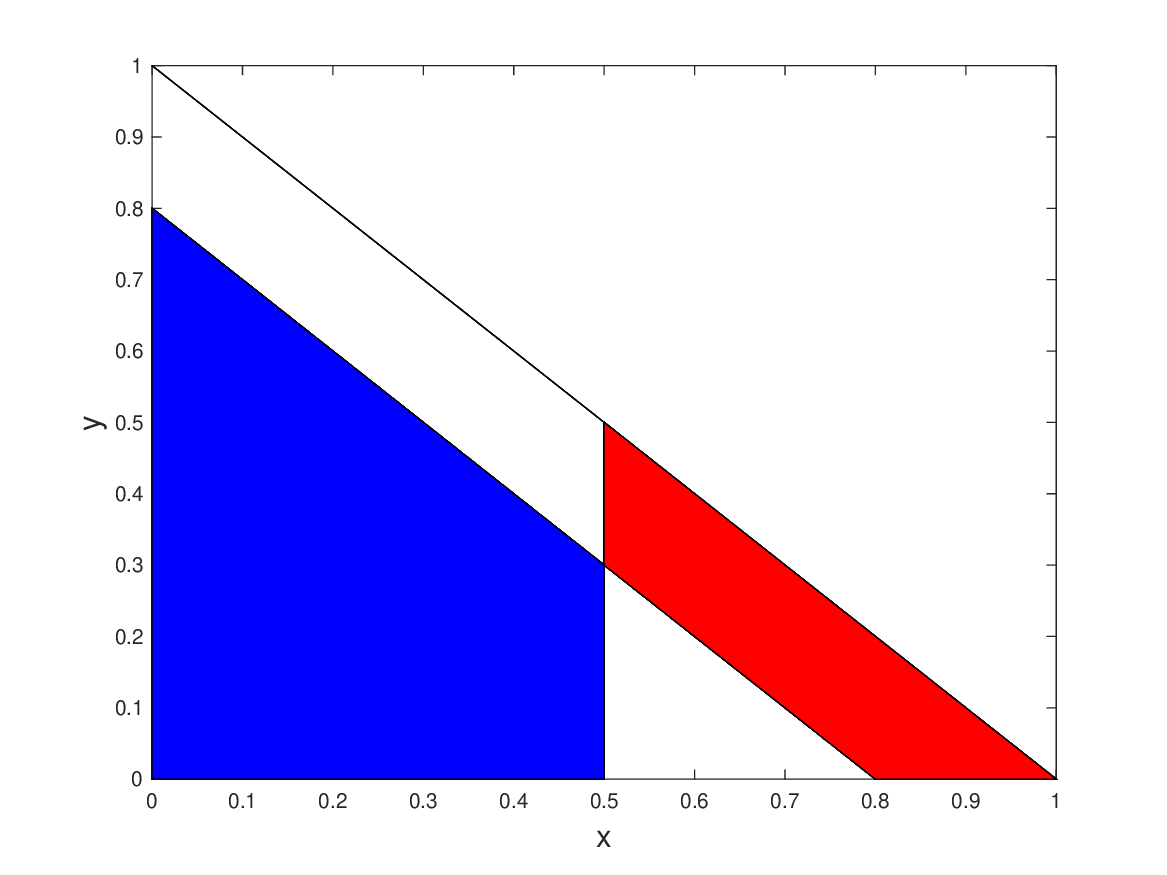}
\caption{(color online). As shown in the figure, the source set (blue), $M_s(\ket{\psi})$, and the accessible set (red), $M_a(\ket{\psi})$, of the state $\rho$ are depicted. Any state in $M_s(\ket{\psi})$ can be transformed to $\ket{\psi}$ via SIO and IC and $\ket{\psi}$ can be transformed into any state in $M_a(\ket{\psi})$ via SIO and IC. In the figure, $\ket{\psi}=\sqrt{0.5}\ket{00}+\sqrt{0.3}\ket{01}+\sqrt{0.2}\ket{10}$.}
\label{fig_qutrit}
\end{figure}

Consider a two-qubit pure state $\ket{\psi}=\sqrt{x_0}\ket{00}+\sqrt{x_1}\ket{01}+\sqrt{x_2}\ket{10}$ with $x_0+x_1+x_2=1$. As shown in Figure~\ref{fig_qutrit}, since every qutrit pure state can be represented as a point in the $x-y$ plane, the accessible volume and source volume are
\begin{align}
V_{a}^{{\mathcal{O}}}(\ket{\psi}) = \frac{1}{2}[(1-x)^2-y^2],
\end{align}

\begin{align}
V_{s}^{{\mathcal{O}}}(\ket{\psi}) = \frac{1}{2}[(x+y)^2-y^2],
\end{align}
where $x\geq y\geq z$, $x,y,z\in\{x_0,x_1,x_2\}$ and $\mathcal{O}\in\{SIO,IC\}$.

The accessible coherence and source coherence are
\begin{align}
C_{a}^{{\mathcal{O}}}(\ket{\psi}) =(1-x)^2-y^2,
\end{align}

\begin{align}
C_{s}^{{\mathcal{O}}}(\ket{\psi})=1-(x+y)^2+y^2.
\end{align}

\end{example}
%%%%%%%%%%%%%%%%%%%%%%%%%%%%%%%%%%%%%%%%%%%%%%%%%%%%%%%%%%%%%%%%%%%%%%%%%%%%%%%%%%%%%%%%%%%%%%%%%%%%%%%%%%%%%%%%%%%%%%%%%%%%%%%%%%%%%%%%%%%%%%%%%%%%%%

%%%%%%%%%%%%%%%%%%%%%%%%%%%%%%%%%%%%%%%%%%%%%%%%%%%%%%%%%%%%%%%%%%%%%%%%%%%%%%%%%%%%%%%%%%%%%%%%%%%%%%%%%%%%%%%%%%%%%%%%%%%%%%%%%%%%%%%%%%%%%%%%%%%%%%

\end{document}